\documentclass[11pt]{article} 
\usepackage[a4paper, margin=2cm]{geometry} 

\usepackage[nolists,nomarkers,tablesfirst]{endfloat} 
 
\usepackage{amssymb,amsmath,amsfonts} 
\usepackage{amsthm} 
\usepackage{esvect,bm} 
\usepackage{graphicx,caption} 
\usepackage[numbers]{natbib} 
\RequirePackage[colorlinks,citecolor=blue,urlcolor=blue]{hyperref}

\newcommand{\e}{{\rm e}} 
\newcommand{\im}{{\rm i}} 
\newcommand{\E}{{\mathbb E}} 
 
\newcommand{\Q}{{\mathbb Q}} 
\newcommand{\C}{{\mathbb C}} 
\newcommand{\R}{{\mathbb R}} 
\newcommand{\N}{{\mathbb N}}

\newcommand{\Ecal}{{\mathcal E}} 
\newcommand{\Fcal}{{\mathcal F}} 
\newcommand{\Gcal}{{\mathcal G}} 
\newcommand{\Hcal}{{\mathcal H}}

\newcommand{\1}{\mathbf{1}} 
 
\newcommand{\dotprod}[2]{(#1,#2)_w} 
 
\DeclareMathOperator{\tr}{Tr} 
\DeclareMathOperator{\diag}{diag} 
\newcommand{\Pol}{{\rm Pol}} 
 
\newtheorem{proposition}{Proposition}[section] 
\newtheorem{lemma}[proposition]{Lemma} 
\newtheorem{theorem}[proposition]{Theorem} 
 
\newtheorem{corollary}[proposition]{Corollary} 
\newtheorem{remark}[proposition]{Remark}

 \def\JELname{{\bfseries JEL Classification}\enspace} 
      \def\JEL#1{\par\addvspace\medskipamount{\rightskip=0pt plus1cm 
      \def\and{\ifhmode\unskip\nobreak\fi\ $\cdot$ 
      }\noindent\JELname\ignorespaces#1\par}} 
 
\title{The Jacobi Stochastic Volatility Model\footnote{We thank the participants at the 2014 Stochastic Analysis in Finance and Insurance Conference in Oberwolfach, the 2015 AMaMeF and Swissquote Conference in Lausanne, the 2016 ICMS Workshop in Edinburgh, and the seminar at Mannheim Mathematics Department, as well as Stefano De Marco, Julien Hugonnier, Wahid Khosrawi-Sardroudi, Martin Larsson, and Peter Tankov for their comments. 
We thank an anonymous referee, an anonymous associate editor, and Chris Rogers (co-editor) for their careful reading of the manuscript and suggestions. 
The research leading to these results has received funding from the European Research Council under the European Union's Seventh Framework Programme (FP/2007-2013) ERC Grant Agreement n.\ 307465-POLYTE. The research of Sergio Pulido benefited from the support of the Chair Markets in Transition (F\'ed\'eration Bancaire Fran\c caise) and the project ANR 11-LABX-0019.}} 
 
\author{Damien Ackerer\footnote{Swissquote Bank, Gland, Switzerland. E-mail:  \href{mailto:damien.ackerer@swissquote.ch}{damien.ackerer@swissquote.ch}}  
\and Damir Filipovi\'c\footnote{EPFL and Swiss Finance Institute, Lausanne, 
Switzerland. E-mail: \href{mailto:damir.filipovic@epfl.ch}{damir.filipovic@epfl.ch}}  
\and Sergio Pulido\footnote{Laboratoire de Math\'ematiques et Mod\'elisation d'\'Evry (LaMME), Universit\'e d'\'Evry-Val-d'Essonne, ENSIIE, Universit\'e Paris-Saclay, \'Evry, France. E-mail: \href{mailto:sergio.pulidonino@ensiie.fr}{sergio.pulidonino@ensiie.fr}} 
} 
 
\date{February 20, 2018} 
 
\begin{document}

\maketitle 
 
\begin{center} 
\large forthcoming in \textit{Finance and Stochastics} 
\end{center} 
 
\bigskip 
 
\begin{abstract} 
We introduce a novel stochastic volatility model where the squared volatility of the asset return follows a Jacobi process. It contains the Heston model as a limit case. We show that the joint density of any finite sequence of log returns admits a Gram--Charlier A expansion with closed-form coefficients. We derive closed-form series representations for option prices whose discounted payoffs are functions of the asset price trajectory at finitely many time points. This includes European call, put, and digital options, forward start options, and can be applied to discretely monitored Asian options. In a numerical analysis we show that option prices can be accurately and efficiently approximated by truncating their series representations.
\end{abstract} 
 
 
 

\newpage 
\section{Introduction}
\label{sec:intro}

Stochastic volatility models for asset returns are popular among practitioners and academics because they can generate implied volatility surfaces that match option price data to a great extent. They resolve the shortcomings of the Black--Scholes model~\cite{black1973pricing}, where the return has constant volatility. Among the the most widely used stochastic volatility models is the Heston model~\cite{heston1993closed}, where the squared volatility of the return follows an affine square-root diffusion. European call and put option prices in the Heston model can be computed using Fourier transform techniques, which have their numerical strengths and limitations; see for instance~\citet{carr1999option}, \citet{bakshi2000spanning}, \citet{duffie2003}, \citet{Fang2009}, and \citet{chen2012generalized}.

In this paper we introduce a novel stochastic volatility model, henceforth the Jacobi model, where the squared volatility $V_t$ of the log price $X_t$ follows a Jacobi process with values in some compact interval $[v_{min},v_{max}]$. As a consequence, Black--Scholes implied volatilities are bounded from below and above by $\sqrt{v_{min}}$ and $\sqrt{v_{max}}$. The Jacobi model $(V_t,X_t)$ belongs to the class of polynomial diffusions studied in~\citet{eri_pis_11}, \citet{cuchiero2012polpres}, and~\citet{filipovic2015polpres}. It includes the Black--Scholes model as a special case and converges weakly in the path space to the Heston model for $v_{max}\to\infty$ and $v_{min}=0$.

We show that the log price $X_T$ has a density $g$ that admits a Gram--Charlier A series expansion with respect to any Gaussian density $w$ with sufficiently large variance. More specifically, the likelihood ratio function $\ell =g/w $ lies in the weighted space $L^2_w$ of square-integrable functions with respect to $w$. Hence it can be expanded as a generalized Fourier series with respect to the corresponding orthonormal basis of Hermite polynomials $H_n(X_0), \, n \ge0$. Boundedness of $V_t$ is essential, as the Gram--Charlier A series of $g$ does not converge for the Heston model. 

The Fourier coefficients $\ell_n$ of $\ell$ are given by the Hermite moments of $X_T$, $\ell_n=\E[H_n(X_T)]$. Due to the polynomial property of $(V_t,X_t)$ the Hermite moments admit easy to compute closed-form expressions. This renders the Jacobi model extremely useful for option pricing. Indeed, the price $\pi_f$ of a European option with discounted payoff $f(X_T)$ for some function $f$ in $L^2_w$ is given by the $L^2_w$-scalar product $\pi_f=\dotprod{f}{\ell}=\sum_{n\ge 0} f_n \ell_n$. The Fourier coefficients $f_n$ of $f$ are given in closed-form for many important examples, including European call, put, and digital options. We approximate $\pi_f$ by truncating the price series at some finite order $N$ and derive truncation error bounds.

We extend our approach to price exotic options whose discounted payoff $f(Y)$ depends on a finite sequence of log returns $Y_i=(X_{t_i}-X_{t_{i-1}}), \,1\le i\le d$. As in the univariate case we derive the Gram--Charlier A series expansion of the density $g$ of $Y$ with respect to a properly chosen multivariate Gaussian density $w$. Assuming that $f$ lies in $L^2_w$ the option price $\pi_f$ is obtained as a series representation of the $L^2_w$-scalar product in terms of the Fourier coefficients of $f$ and of the likelihood ratio function $\ell=g/w$ given by the corresponding Hermite moments of $Y$. Due to the polynomial property of $(V_t,X_t)$ the Hermite moments admit closed-form expressions, which can be efficiently computed. The Fourier coefficients of $f$ are given in closed-form for various examples, including forward start options and forward start options on the underlying return.

Consequently, the pricing of these options is extremely efficient and does not require any numerical integration. Even when the Fourier coefficients of the discounted payoff function $f$ are not available in closed-form, e.g.\ for Asian options, prices can be approximated by integrating $f$ with respect to the Gram--Charlier A density approximation of $g$. This boils down to a numerically feasible integration with respect to the underlying Gaussian density $w$. In a numerical analysis we find that the price approximations become accurate within short CPU time. This is in contrast to the Heston model, for which the pricing of exotic options using Fourier transform techniques is cumbersome and creates numerical difficulties as reported in~\citet{kruse2005pricing}, \citet{kahl2005not}, and \citet{albrecher2006little}. In view of this, the Jacobi model also provides a viable alternative to approximate option prices in the Heston model.

The Jacobi process, also known as Wright--Fisher diffusion, was originally used to model gene frequencies;~see for instance~\citet{karlin1981second} and~\citet{ethierkurtz86}. More recently, the Jacobi process has also been used to model financial factors. For example, \citet{delbaen2002interest} model interest rates by the Jacobi process and study moment-based techniques for pricing bonds. In their framework, bond prices admit a series representation in terms of Jacobi polynomials. These polynomials constitute an orthonormal basis of eigenfunctions of the infinitesimal generator and the stationary beta distribution of the Jacobi process; additional properties of the Jacobi process can be found in~\citet{mazet97} and~\citet{demni2009large}. 
The multivariate Jacobi process has been studied in~\citet{gourieroux2006multivariate} where the authors suggest it to model smooth regime shifts and give an example of stochastic volatility model without leverage effect.
The Jacobi process has been also applied recently to model stochastic correlation matrices in~\citet{AhdidaAlfonsi13} and credit default swap indexes in~\citet{bernisscotti16}.

Density series expansion approaches to option pricing were pioneered by \citet{jarrow1982approximate}. They propose expansions of option prices that can be interpreted as corrections to the pricing biases of the Black--Scholes formula. They study density expansions for the law of underlying prices, not the log returns, and express them in terms of cumulants. Evidently, since convergence cannot be guaranteed in general, their study is based on strong assumptions that imply convergence. In subsequent work, \citet{corrado1996skewness} and \citet{Corrado97impliedvolatility} study Gram--Charlier A expansions of 4$^\text{th}$ order for options on the S\&P 500 index. These expansions contain skewness and kurtosis adjustments to option prices and implied volatility with respect to the Black--Scholes formula. The skewness and kurtosis correction terms, which depend on the cumulants of $3^{\text{rd}}$ and $4^{\text{th}}$ order, are estimated from data. Due to the instability of the estimation procedure, higher order expansions are not studied. Similar studies on the biases of the Black--Scholes formula using Gram--Charlier A expansions include \citet{backus2004} and~\citet{limelnikov12}. More recently, \citet{drimus2013closed} and~\citet{necula2015general} study related expansions with Hermite polynomials. In order to guarantee the convergence of the Gram--Charlier A expansion for a general class of diffusions,~\citet{Ait-Sahalia2002} develop a technique based on a suitable change of measure. As pointed out in~\citet{filipovic2013density}, in the affine and polynomial settings this change of measure usually destroys the polynomial property and the ability to calculate moments efficiently. More recently a similar study has been carried out by~\citet{Xiu2014}. Gram--Charlier A expansions, under a change of measure, are also mentioned in the work of~\citet{MadanMilne94}, and the subsequent studies of~\citet{Longstaff95}, \citet{AbkenMadanRamamurtie96} and~\citet{brenner-eom97}, where they use these moment expansions to test the martingale property with financial data and hence the validity of a given model.

Our paper is similar to \citet{filipovic2013density} in that it provides a generic framework to perform density expansions using orthonormal polynomial basis in weighted $L^2$ spaces for affine models. They show that a bilateral Gamma density weight works for the Heston model. However, that expansion is numerically more cumbersome than the Gram--Charlier A expansion because the orthonormal basis of polynomials has to be constructed using Gram--Schmidt orthogonalization. In a related paper \citet{heston2016spanning} study polynomial expansions of prices in the Heston, Hull-White and Variance Gamma models using logistic weight functions.

The remainder of the paper is as follows. In Section~\ref{sec:model} we introduce the Jacobi stochastic volatility model. In Section~\ref{sec:optionprice} we derive European option prices based on the Gram--Charlier A series expansion. In Section~\ref{sec:exotic} we extend this to the multivariate case, which forms the basis for exotic option pricing and contains the European options as special case. In Section~\ref{sec:numeric} we give some numerical examples. In Section~\ref{sec:conclusions} we conclude. In Appendix~\ref{appHm} we explain how to efficiently compute the Hermite moments. All proofs are collected in Appendix~\ref{a:proofs}.


\section{Model specification}
\label{sec:model}

We study a stochastic volatility model where the squared volatility follows a Jacobi process. Fix some real parameters $0\leq v_{min}< v_{max}$, and define the quadratic function
	\[ Q(v) = \frac{(v-v_{min})(v_{max}-v)}{(\sqrt{v_{max}}-\sqrt{v_{min}})^2} .\]
Inspection shows that $v\ge Q(v)$, with equality if and only if $v=\sqrt{v_{min}v_{max}}$, and $Q(v)\ge 0$ for all $v\in [v_{min},v_{max}]$, see Figure~\ref{fig:varcor} for an illustration.

We consider the diffusion process $(V_t,X_t)$ given by
	\begin{equation}\label{sdeXV}
	\begin{aligned}
	dV_t &= \kappa(\theta   - V_t)\,dt + \sigma\sqrt{Q(V_t)}\,dW_{1t}\\
	dX_t &= \left(r-\delta-V_t/2\right)dt + \rho\, \sqrt{Q(V_t)}\,dW_{1t}  + \sqrt{V_t -\rho^2 \,Q(V_t)}\,dW_{2t}
	\end{aligned}
	\end{equation}
for real parameters $\kappa> 0$, $\theta\in (v_{min},v_{max}]$, $\sigma> 0$, interest rate $r$, dividend yield $\delta$, and $\rho\in [-1,1]$, and where $W_{1t}$ and $W_{2t}$ are independent standard Brownian motions on some filtered probability space $(\Omega,\Fcal,\Fcal_t,\Q)$. The following theorem shows that $(V_t,X_t)$ is well defined.

\begin{theorem}\label{thmexiuni}
For any deterministic initial state $(V_0,X_0)\in  [v_{min},v_{max}]\times \R$ there exists a unique solution $(V_t,X_t)$ of \eqref{sdeXV} taking values in $[v_{min},v_{max}]\times \R$ and satisfying
\begin{equation}\label{Vtvstart}
  \int_0^\infty \1_{\{V_t=v\}} dt = 0\quad\text{for all $v\in [v_{min},v_{max})$.}
\end{equation}
Moreover, $V_t$ takes values in $(v_{min},v_{max})$ if and only if $V_0\in  (v_{min},v_{max})$ and
	\begin{equation}\label{eq:boundary}
	\frac{\sigma^2(v_{max}-v_{min})}{(\sqrt{v_{max}}-\sqrt{v_{min}})^2}\leq 2\kappa\min\{v_{max}-\theta, \theta-v_{min}\}.
	\end{equation}
\end{theorem}

\begin{remark}
Property~\eqref{Vtvstart} implies that no state $v\in [v_{min},v_{max})$ is absorbing. It also implies that conditional on $\{ V_{t},\, t\in [0,T]\}$, the increments $X_{t_i}-X_{t_{i-1}}$ are non-degenerate Gaussian for any $t_{i-1}<t_i\le T$ as will be shown in the proof of Theorem~\ref{thmdensitymd}. Taking $v_{min}=0$ and the limit as $v_{max}\to\infty$, condition~\eqref{eq:boundary} coincides with the known condition that precludes the zero lower bound for the CIR process, $\sigma^2\leq 2\kappa\theta$.
\end{remark}

We specify the price of a traded asset by $S_t = \e^{X_t}$. Then $\sqrt{V_t}$ is the stochastic volatility of the asset return, $d\langle X,X \rangle_t = V_t\,dt $. The cumulative dividend discounted price process $\e^{-(r-\delta)t} S_t$ is a martingale. In other words, $\Q$ is a risk-neutral measure. The parameter $\rho$ tunes the instantaneous correlation between the asset return and the squared volatility,
	\[\frac{d\langle V,X\rangle_t}{\sqrt{d\langle V,V\rangle_t}\sqrt{d\langle X,X\rangle_t}} = \rho\,\sqrt{Q(V_t)/V_t}.\]
This correlation is equal to $\rho$ if $V_t=\sqrt{v_{min}v_{max}}$, see Figure~\ref{fig:varcor}. In general, we have $\sqrt{Q(V_t)/V_t}\le 1$. Empirical evidences suggest that $\rho$ is negative when $S_t$ is a stock price or index. This is commonly referred as the leverage effect, that is, an increase in volatility often goes along with a decrease in asset value.

Since the instantaneous squared volatility $V_t$ follows a bounded Jacobi process on the interval $[v_{min},v_{max}]$, we refer to \eqref{sdeXV} as the \textit{Jacobi model.} For $V_{0}=\theta=v_{max}$ we have constant volatility $V_t=V_0$ for all $t\geq 0$ and we obtain the Black--Scholes model
	\begin{equation}\label{BSmodel}
	dX_t= \left(r-\delta-V_0/2\right)dt+ \sqrt{V_0}\,dW_{2t}.
	\end{equation}
For $v_{min}= 0$ and the limit $v_{max}\to\infty$ we have $Q(v)\to v$, and we formally obtain the Heston model as limit case of \eqref{sdeXV},
	\begin{equation}\label{sdeHeston}
	\begin{aligned}
	dV_t &= \kappa(\theta   - V_t)\,dt + \sigma\sqrt{V_t}\,dW_{1t}\\
	dX_t &= \left(r-\delta-V_t/2\right)dt + \sqrt{V_t}\left(\rho \,dW_{1t}  + \sqrt{(1-\rho^2)}\,dW_{2t}\right).
	\end{aligned}
	\end{equation}
	
In fact, the Jacobi model \eqref{sdeXV} is robust with respect to perturbations, or mis-specifications, of the model parameters $v_{min}$, $v_{max}$ and initial state $(V_0,X_0)$. Specifically, the following theorem shows that the diffusion \eqref{sdeXV} is weakly continuous in the space of continuous paths with respect to  $v_{min}$, $v_{max}$ and $(V_0,X_0)$. In particular, the Heston model \eqref{sdeHeston} is indeed a limit case of our model~\eqref{sdeXV}.

Consider a sequence of parameters $0\le v_{min}^{(n)}<v_{max}^{(n)} $ and deterministic initial states $(V_0^{(n)},X_0^{(n)}) \in [v_{min}^{(n)},v_{max}^{(n)}]\times\R$  converging to ${0\le v_{min}<v_{max}\le \infty}$ and $(V_0,X_0)\in [0,\infty)\times\R$ as $n\to\infty$, respectively. We denote by $(V_t^{(n)},X_t^{(n)})$ and $(V_t,X_t)$ the respective solutions of~\eqref{sdeXV}, or \eqref{sdeHeston} if $v_{max}=\infty$. Here is our main convergence result.

\begin{theorem}\label{thmconv}
The sequence of diffusions $(V_t^{(n)},X_t^{(n)})$ converges weakly in the path space to $(V_t,X_t)$ as $n\to\infty$.
\end{theorem}

As the discounted put option payoff function $f_{put}(x) = \e^{-rT} (\e^k - \e^x)^+$ is bounded and continuous on $\R$, it follows from the weak continuity stated in Theorem~\ref{thmconv} that the put option prices based on $(V_t^{(n)},X_t^{(n)})$ converge to the put option price based on the limiting model $(V_t,X_t)$ as $n\to\infty$. The put-call parity, $\pi_{call}  - \pi_{put}  = \e^{-\delta T}S_0-\e^{-rT+k} $, then implies that also call option prices converge as $n\to\infty$. This carries over to more complex path-dependent options with bounded continuous payoff functional.

\subsection*{Polynomial property}

Moments in the Jacobi model~\eqref{sdeXV} are given in closed-form. Indeed, let
	\[ \Gcal f(v,x)=b(v)^\top\nabla f(v,x)+\frac{1}{2}\tr \left(a(v)\nabla^2f(v,x)\right)\]
denote the generator of $(V_t,X_t)$ with drift vector $b(v)$ and the diffusion matrix $a(v)$ given by
	\begin{equation}\label{eqgenerator}
	b(v) =\begin{pmatrix}\kappa(\theta-v)\\r-\delta-v/2\end{pmatrix},\quad  a(v)=\begin{pmatrix} \sigma^2 Q(v) & \rho\sigma Q(v) \\ \rho\sigma Q(v) & v\end{pmatrix}.
	\end{equation}
Observe that $a(v)$ is  continuous in the parameters $v_{min}$, $v_{max}$, so that for $v_{min}=0$ and $v_{max}\to\infty$ we obtain
	\[ a(v)\to \begin{pmatrix} \sigma^2 v & \rho\sigma v \\ \rho\sigma v & v\end{pmatrix},\]
which corresponds to the generator of the Heston model~\eqref{sdeHeston}. Let $\Pol_n$ be the vector space of polynomials in $(v,x)$ of degree less than or equal to $n$. It then follows by inspection that the components of $b(v)$ and $a(v)$ lie in $\Pol_1$ and $\Pol_2$, respectively. As a consequence, $\Gcal$ maps any polynomial of degree $n$ onto a polynomial of degree $n$ or less, $\Gcal\, \Pol_n\subset\Pol_n$, so that $(V_t,X_t)$ is a polynomial diffusion, see~\citet[Lemma~2.2]{filipovic2015polpres}. From this we can easily calculate the conditional moments of $(V_T,X_T)$ as follows. For $N\in\N$, let $M=(N+2)(N+1)/2$ denote the dimension of $\Pol_N$. Let $h_1(v,x),\ldots,h_{M}(v,x)$ be a basis of polynomials of $\Pol_N$ and denote by $G$ the matrix representation of the linear map $\Gcal$ restricted to $\Pol_N$ with respect to this basis.
\begin{theorem}\label{thmoments}
 For any polynomial $p\in\Pol_N$ and $0\le t\le T$ we have
	\[ \E\big[p(V_T,X_T)\bigm| \Fcal_t\big]=\begin{pmatrix}h_1(V_t,X_t) & \cdots & h_{M}(V_t,X_t)\end{pmatrix}\e^{(T-t)G}\vv p \]
where $\vv p\in\R^{M}$ is the coordinate representation of the polynomial $p(v,x)$ with respect to the basis $h_1(v,x),\ldots,h_{M}(v,x)$.
\end{theorem}
\begin{proof}
See \citet[Theorem 3.1]{filipovic2015polpres}.
\end{proof}
The moment formula in Theorem~\ref{thmoments} is crucial in order to efficiently implement the numerical schemes described below.


\section{European option pricing}
\label{sec:optionprice}

Henceforth we assume that $(V_0,X_0)\in  [v_{min},v_{max}]\times \R$ is a deterministic initial state and fix a finite time horizon $T>0$. We first establish some key properties of the distribution of $X_T$. Denote the quadratic variation of the second martingale component of $X_t$ in~\eqref{sdeXV} by
	\begin{equation}\label{eq:defC_T}
	{C_t}=\int_0^t \left(V_s -\rho^2 Q(V_s)\right) ds.
	\end{equation}
	
The following theorem is a special case of Theorem~\ref{thmdensitymd} below.
\begin{theorem}\label{thmdensity}
Let $\epsilon<1/(2 v_{max} T)$. The distribution of $X_T$ admits a density $g_T(x)$ on $\R$ that satisfies
	\begin{equation}\label{gint}
	\int_\R \e^{\epsilon x^2 } g_T(x)\,dx<\infty .
	\end{equation}
If
	\begin{equation}\label{assint}
	\E\left[  {C_T}^{-1/2-k} \right]<\infty
	\end{equation}
for some $k\in\N_0$ then $g_T(x)$ and $\e^{\epsilon x^2 } g_T(x)$ are uniformly bounded and $g_T(x)$ is $k$-times continuously differentiable on $\R$. A sufficient condition for~\eqref{assint} to hold for any $k\ge 0$ is
	\begin{equation}\label{assint_2}
	v_{min}>0\,\text{ and }\,\rho^2<1.\footnote{We conjecture that \eqref{assint} holds for any $k\ge 0$ also when $v_{min}=0$ (and $\kappa\theta>0$) or $\rho^2=1$. For the Heston model~\eqref{sdeHeston} with $Q(v)=v$ and $\rho^2<1$ the conjecture follows from \citet[Theorem~4.1]{duf_01}.}
	\end{equation}
\end{theorem}

The condition that $\epsilon<1/(2 v_{max} T)$ is sharp for \eqref{gint} to hold. Indeed, consider the Black--Scholes model~\eqref{BSmodel} where $V_t=\theta=v_{max}$ for all $t\ge 0$. Then $X_T$ is Gaussian with variance $C_T=v_{max}T$. Hence the integral in \eqref{gint} is infinite for any $\epsilon\ge 1/(2 v_{max} T)$.

Since any uniformly bounded and integrable function on $\R$ is square integrable on $\R$, as an immediate consequence of Theorem~\ref{thmdensity} we have the following corollary.

\begin{corollary}\label{corL2w}
Assume \eqref{assint} holds for $k=0$. Then
	\begin{equation}\label{eq:ginL2}
	\int_\R \frac{g_T(x)^2}{w(x)}\,dx <\infty
	\end{equation}
for any Gaussian density $w(x)$ with variance $\sigma_w^2$ satisfying
	\begin{equation}\label{sicon}
	\sigma_w^2 >\frac{v_{max}T}{2}.
	\end{equation}
\end{corollary}

\begin{remark}\label{remHESTON}
It follows from the proof that the statements of Theorem~\ref{thmdensity} also hold for the Heston model~\eqref{sdeHeston} with $Q(v)=v$ and $\epsilon = 0$. However, the Heston model does not satisfy \eqref{gint} for any $\epsilon>0$. Indeed, otherwise its moment generating function
\begin{equation}\label{eqmmf}
\widehat{g_T}(z)=\int_\R \e^{ zx}g_T(x)\,dx
\end{equation}
would extend to an entire function in $z\in\C$. But it is well known that $\widehat{g_T}(z)$ becomes infinite for large enough $z\in\R$, see~\citet{andersen2007moment}. As a consequence, the Heston model does not satisfy~\eqref{eq:ginL2} for any finite $\sigma_w$. Indeed, by the Cauchy-Schwarz inequality,~\eqref{eq:ginL2} implies~\eqref{gint} for any $\epsilon<1/(4\sigma_w^2)$.
\end{remark}

We now compute the price at time $t=0$ of a European claim with discounted payoff $f(X_T)$ at expiry date $T>0$. We henceforth assume that \eqref{assint} holds with $k=0$, and we let $w(x)$ be a Gaussian density with mean $\mu_w$ and variance $\sigma_w^2$ satisfying \eqref{sicon}. We define the weighted Lebesgue space
	\[ L^2_w=\left\{ f(x) : \| f\|_w^2 = \int_\R f(x)^2 \,w(x)dx<\infty\right\},\]
which is a Hilbert space with scalar product
	\[ \dotprod{f}{g} = \int_\R f(x)g(x)\,w(x)dx .\]
The space $L^2_w$ admits the orthonormal basis of generalized Hermite polynomials $H_n(x)$, $n\ge 0$, given by
	\begin{equation}\label{HHcal}
 	H_n(x) = \frac{1}{\sqrt{n!}} \Hcal_n \left(\frac{x-\mu_w}{\sigma_w}\right)
	\end{equation}
where $\Hcal_n(x)$ are the standard Hermite polynomials defined by
\begin{equation}\label{PHdef}
 \Hcal_n(x)=(-1)^n\e^{\frac{x^2}{2}}\frac{d^n}{d x^n}\e^{-\frac{x^2}{2}},
\end{equation}
see \citet[Section XVI.1]{feller1960introduction}. In particular, the degree of $H_n(x)$ is $n$, and ${\dotprod{H_m}{H_n}=1}$ if $m=n$ and zero otherwise.

Corollary~\ref{corL2w} implies that the likelihood ratio function $\ell(x)=g_T(x)/w(x)$ of the density $g_T(x)$ of the log price $X_T$ with respect to $w(x)$ belongs to $L^2_w$. We henceforth assume that also the discounted payoff function $f(x)$ is in $L^2_w$. This hypothesis is satisfied for instance in the case of European call and put options. It implies that the price, denoted by $\pi_f$, is well defined and equals
	\begin{equation}\label{pihinf}
	\pi_f = \int_\R f(x) g_T(x)\,dx = \dotprod{f}{\ell}  =\sum_{n\ge 0} f_n \ell_n,
	\end{equation}
for the \textit{Fourier coefficients} of $f(x)$
	\begin{equation}\label{eq:Fcoef}
	f_n=\dotprod{f}{H_n},
	\end{equation}
and the Fourier coefficients of $\ell(x)$ that we refer to as \textit{Hermite moments}
	\begin{equation}\label{eq:Hmoments}
	\ell_n=\dotprod{\ell}{H_n}=\int_\R H_n(x) g_T(x)\,dx.
	\end{equation}

We approximate the price $\pi_f$ by truncating the series in \eqref{pihinf} at some order $N\ge 1$ and write
	\begin{equation}\label{eqproxyprice}
	\pi_f^{(N)} = \sum_{n= 0}^N f_n \ell_n,
	\end{equation}
so that $\pi_f^{(N)}\to \pi_f$ as $N\to\infty$. Due to the polynomial property of the Jacobi model, \eqref{eqproxyprice} induces an efficient price approximation scheme because the Hermite moments $\ell_n$ are linear combinations of moments of $X_T$ and thus given in closed-form, see Theorem~\ref{thmoments}. In particular, since $H_0(x)=1$, we have $\ell_0=1$. More details on the computation of $\ell_n$ are given in Appendix~\ref{appHm}.

With the Hermite moments $\ell_n$ available, the computation of the approximation \eqref{eqproxyprice} boils down to a numerical integration,
	\begin{equation}\label{eqpihNint}
  	\pi_f^{(N)} = \sum_{n=0}^N \dotprod{f}{\ell_n H_n}=\int_\R f(x) \ell^{(N)}(x) \,w(x)dx ,
	\end{equation}
of $f(x)\ell^{(N)}(x)$ with respect to the Gaussian distribution $w(x)dx$, where the polynomial $\ell^{(N)}(x)=\sum_{n=0}^N \ell_n H_n(x)$ is in closed-form. The integral \eqref{eqpihNint} can be computed by quadrature or Monte-Carlo simulation. In specific cases, we find closed-form formulas for the Fourier coefficients $f_n$ and no numerical integration is needed. This includes European call, put, and digital options, as shown below.

\begin{remark}\label{R:divHeston}
Formula~\eqref{eqpihNint} shows that $g_T^{(N)}(x) = \ell^{(N)}(x)w(x)$ serves as an approximation for the density $g_T(x)$. In fact, we readily see that $g_T^{(N)}(x)$ integrates to one and converges to $g_T(x)$ in $L^2_{1/w}$ as $N\to\infty$. Hence, we have convergence of the Gram--Charlier A series expansion of the density of the log price $X_T$ in $L^2_{1/w}$.\footnote{A Gram--Charlier A series expansion of a density function $g(x)$ is formally defined as $g(x)=\sum_{n\ge 0} c_n H_n(x) w(x)$ for some real numbers $c_n$, $n\ge0$.} In view of Remark~\ref{remHESTON}, this does not hold for the Heston model.
\end{remark}

Matching the first moment or the first two moments of $w(x)$ and $g_T(x)$, we further obtain
\begin{equation*}
\ell_1=\int_\R H_1(x) g_T(x)\,dx= \dotprod{H_0}{H_1} = 0\quad\text{if $\mu_w=\E[X_T]$,}\\
\end{equation*}
and similarly,
\begin{equation}\label{mm2}
\ell_1=\ell_2=0\quad\text{if $\mu_w=\E[X_T]$ and $\sigma_w^2={\rm var}[X_T]$.}
\end{equation}
Matching the first moment or the first two moments of $w(x)$ and $g_T(x)$ can improve the convergence of the approximation \eqref{eqproxyprice}. Note however that \eqref{sicon} and \eqref{mm2} imply ${\rm var}[X_T]>v_{max}T/2$, so that second moment matching is not always feasible in empirical applications.

\begin{remark}\label{remBSoptionP}
If $\mu_w=X_0+(r-\delta)T-\sigma_w^2/2$, then $f_0=\int_{\R} f(x) w(x)dx$ is the Black--Scholes option price with volatility parameter $\sigma_{BS}=\sigma_w/\sqrt{T}$. Because ${\E[X_T]=X_0+(r-\delta)T-{\rm var}[X_T]/2}$, this holds in particular if the first two moments of $w(x)$ and $g_T(x)$ match, see~\eqref{mm2}. In this case, the higher order terms in~$\pi_f^{(N)}=f_0+\sum_{n=3}^N f_n\ell_n$ can be thought of as corrections to the corresponding Black--Scholes price $f_0$ due to stochastic volatility.
\end{remark}

The following result, which is a special case of Theorem~\ref{thm:IVEx} below, provides universal upper and lower bounds on the implied volatility of a European option with discounted payoff $f(X_T)$ at $T$ and price $\pi_f$. The implied volatility $\sigma_{\rm IV}$ is defined as the volatility parameter that renders the corresponding Black--Scholes option price equal to $\pi_f$.

\begin{theorem}\label{thm:ivB}
Assume that the discounted payoff function $f(\log(s))$ is convex in $s>0$. Then the implied volatility satisfies $\sqrt{v_{min}} \le \sigma_{\rm IV} \le  \sqrt{v_{max}}$.
\end{theorem}

\subsection*{Examples}

We now present examples of discounted payoff functions $f(x)$ for which closed-form formulas for the Fourier coefficients $f_n$ exist. The first example is a call option.\footnote{Similar recursive relations of the Fourier coefficients for the physicist Hermite polynomial basis can be found in \citet{drimus2013closed}. The physicist Hermite polynomial basis is the orthogonal polynomial basis of the $L^2_w$ space equipped with the weight function $w(x)=\e^{-x^2}$ so that $\dotprod{H_n}{H_n}=\sqrt{2\pi}2^nn!$.}

\begin{theorem}\label{thmoptionFC}
Consider the discounted payoff function for a call option with log strike $k$,
	\begin{equation}\label{Eurcall}
  	f(x) = \e^{-rT} \left( \e^x-\e^k \right)^+ .
	\end{equation}
Its Fourier coefficients $f_n$ in \eqref{eq:Fcoef} are given by
	\begin{equation}\label{E:fourier_coef} \begin{split} f_0&=\e^{-rT+\mu_w}I_0\left(\frac{k-\mu_w}{\sigma_w};\sigma_w\right)-\e^{-rT+k}\Phi\left(\frac{\mu_w-k}{\sigma_w}\right);\\
f_n&=\e^{-rT+\mu_w}\frac{1}{\sqrt{n!}}\sigma_wI_{n-1}\left(\frac{k-\mu_w}{\sigma_w};\sigma_w\right),\quad n\ge 1.
	\end{split}
	\end{equation}
The functions $I_n(\mu;\nu)$ are defined recursively by
	\begin{equation}\label{E:recursionI}
	\begin{aligned}
	I_0(\mu;\nu)&=\e^{\frac{\nu^2}{2}} \Phi(\nu-\mu);\\
    	I_n(\mu;\nu)&= \Hcal_{n-1}(\mu)\e^{\nu \mu}\phi(\mu)+\nu I_{n-1}(\mu;\nu) ,\quad n\ge 1,
	\end{aligned}
	\end{equation}
where $\Phi(x)$ denotes the standard Gaussian distribution function and $\phi(x)$ its density.
\end{theorem}

The Fourier coefficients of a put option can be obtained from the put-call parity. For digital options, the Fourier coefficients $f_n$ are as follows.
\begin{theorem}\label{thmoptionDigital}
Consider the discounted payoff function for a digital option of the form
	\[f(x)=\e^{-rT}\1_{[k,\infty)}(x).\]
Its Fourier coefficients $f_n$ are given by
	\begin{equation}\label{E:fourier_coef_dig}
	\begin{split}
	f_0&=\e^{-rT}\Phi\left(\frac{\mu_w-k}{\sigma_w}\right);\\
	f_n&=\frac{{\rm e}^{-rT}}{\sqrt{n!}}\Hcal_{n-1}\left(\frac{k-\mu_w}{\sigma_w}\right)\phi\left(\frac{k-\mu_w}{\sigma_w}\right),\quad n\ge1,	
	\end{split}
	\end{equation}
where $\Phi(x)$ denotes the standard Gaussian distribution function and $\phi(x)$ its density.
\end{theorem}

For a digital option with generic payoff $\1_{[k_1,k_2)}(x)$ the Fourier coefficients can be derived using Theorem~\ref{thmoptionDigital} and $\1_{[k_1,k_2)}(x)=\1_{[k_1,\infty)}(x)-\1_{[k_2,\infty)}(x).$

\subsection*{Error bounds and asymptotics}

We first discuss an error bound of the price approximation scheme~\eqref{eqproxyprice}. The error of the approximation is ${\epsilon^{(N)}=\pi_f-\pi_f^{(N)}=\sum_{n=N+1}^\infty f_n \ell_n}$ for a fixed order $N\ge 1$. The Cauchy--Schwarz inequality implies the following error bound
\begin{equation}\label{BAERR}
\lvert\epsilon^{(N)} \rvert \le  \left( \| f \|_w^2 - \sum_{n=0}^N f_n^2 \right)^{\frac{1}{2}}  \left( \| \ell \|_w^2 - \sum_{n=0}^N \ell_n^2 \right)^{\frac{1}{2}}.
\end{equation}
The $L^2_w$-norm of $f(x)$ has an explicit expression, $\|f\|_w^2 = \int_\R f(x)^2 \, w(x)dx $, that can be computed by quadrature or Monte--Carlo simulation. The Fourier coefficients $f_n$ can be computed similarly. The Hermite moments $\ell_n$ are given in closed-form. It remains to compute the $L^2_w$-norm of $\ell(x)$. For further use we define
\begin{equation}\label{eq:defM_Tnew}
	M_t=X_0+\int_0^t \left(r-\delta-V_s/2\right)ds +\frac{\rho}{\sigma}\left( V_t-V_0 -\int_0^t \kappa\left(\theta-V_s\right)ds \right),
\end{equation}
so that, in view of \eqref{sdeXV}, the log price $X_t=  M_t + \int_0^t \sqrt{V_s -\rho^2 Q(V_s)}\,dW_{2s}$. Recall also $C_t$ given in~\eqref{eq:defC_T}.

\begin{lemma}\label{lemellnorm}
The $L^2_w$-norm of $\ell(x)$ is given by
\begin{equation}\label{eqellExp}
\|\ell\|_w^2  = \int_\R \frac{g_T(x)^2 }{w(x)}dx  = \E\left[ \frac{g_T(X_T) }{w(X_T)} \right]   =\E\left[ \frac{ \phi\left(X_T,\widetilde{M}_T,{\widetilde{C}_T}\right)}{\phi\left(X_T,\mu_w,\sigma_w^2\right)} \right]
\end{equation}
where $\phi(x,\mu,\sigma^2)$ is the normal density function in $x$ with mean $\mu$ and variance $\sigma^2$, and the pair of random variables $(\widetilde{M}_T,\widetilde{C}_T)$ is independent from $X_T$ and has the same distribution as $(M_T,C_T)$.
\end{lemma}

In applications, we compute the right hand side of \eqref{eqellExp} by Monte--Carlo simulation of $(X_T,\widetilde{M}_T,\widetilde{C}_T)$ and thus obtain the error bound \eqref{BAERR}.

We next show that the Hermite moments $\ell_n$ decay at an exponential rate under some technical assumptions.

\begin{lemma}\label{lemdecayH}
Suppose that \eqref{assint_2} holds and $\sigma_w^2>v_{max}T$. Then there exist finite constants $C>0$ and $0<q<1$ such that $\ell_n^2\leq Cq^n$ for all $n\geq 0$.
\end{lemma}

\subsection*{Comparison to Fourier transform}

An alternative dual expression of the price $\pi_f$ in \eqref{pihinf} is given by the Fourier integral
\begin{equation}\label{pihFourier}
  \pi_f = \frac{1}{2\pi}\int_{\R} \hat f(-\mu-\im\lambda) \hat g_T(\mu+\im\lambda)d\lambda,
\end{equation}
	
where $\widehat f(z)$ and $\widehat{g_T}(z)$ denote the moment generating functions given by \eqref{eqmmf}, respectively. Here $\mu\in\R$ is some appropriate dampening parameter such that $\e^{-\mu x}f(x)$ and $\e^{\mu x}g_T(x)$ are Lebesgue integrable and square integrable on $\R$. Indeed, Lebesgue integrability implies that $\widehat f(z)$ and $\widehat{g_T}(z)$ are well defined for $z\in\mu+\im\R$ through \eqref{eqmmf}. Square integrability and the Plancherel Theorem then yield the representation~\eqref{pihFourier}. 
For example, for the European call option \eqref{Eurcall} we have $\widehat f(z)=\e^{-rT+k(1+z)}/(z(z+1))$ for ${\rm Re}(z)<-1$

Option pricing via \eqref{pihFourier} is the approach taken in the Heston model \eqref{sdeHeston}, for which there exists a closed-form expression for $\widehat{g_T}(z)$. It is given in terms of the solution of a Riccati equation. The computation of $\pi_f$ boils down to the numerical integration of \eqref{pihFourier} along with the numerical solution of a Riccati equation for every argument $z\in\mu+\im\R$ that is needed for the integration. The Heston model (which entails $v_{max}\to\infty$) does not adhere to the series representation \eqref{pihinf} that is based on condition~\eqref{eq:ginL2}, see Remark~\ref{remHESTON}.

The Jacobi model, on the other hand, does not admit a closed-form expression for $\widehat{g_T}(z)$. But the Hermite moments $\ell_n$ are readily available in closed-form. In conjunction with Theorem~\ref{thmoptionFC}, the (truncated) series representation~\eqref{pihinf} thus provides a valuable alternative to the (numerical) Fourier integral approach~\eqref{pihFourier} for option pricing. Moreover, the approximation~\eqref{eqpihNint} can be applied to any discounted payoff function $f(x)\in L^2_w$. This includes functions $f(x)$ that do not necessarily admit closed-form moment generating function $\widehat f(z)$ as is required in the Heston model approach. In Section~\ref{sec:exotic}, we further develop our approach to price path dependent options, which could be a cumbersome task using Fourier transform techniques in the Heston model.


\section{Exotic option pricing}
\label{sec:exotic}

Pricing exotic options with stochastic volatility models is a challenging task. We show that the price of an exotic option whose payoff is a function of a finite sequence of log returns admits a polynomial series representation in the Jacobi model.

Henceforth we assume that $(V_0,X_0)\in  [v_{min},v_{max}]\times \R$ is a deterministic initial state. Consider time points $0=t_0< t_1<t_2<\cdots<t_d$ and denote the log returns $Y_{t_i}=X_{t_i}-X_{t_{i-1}}$ for $i=1,\ldots,d$. The following theorem contains Theorem~\ref{thmdensity} as special case where $d=1$.

\begin{theorem}\label{thmdensitymd}
Let $\epsilon_1,\ldots,\epsilon_d\in\R$ be such that $\epsilon_i< 1/(2 v_{max} (t_i-t_{i-1}))$ for ${i=1,\ldots,d}$. The random vector $(Y_{t_1},\ldots,Y_{t_d})$ admits a density ${g_{t_1,\ldots,t_d}(y)}$ on $\R^d$ satisfying
	\begin{equation*}\label{gint-md}
	\int_{\R^d} \e^{\sum_{i=1}^d\epsilon_i y_i^2 } g_{t_1,\ldots,t_d}(y)\,dy<\infty.
	\end{equation*}
If
	\begin{equation}\label{assintmd}
	\E\left[  \prod_{i=1}^d ({C_{t_{i}}}-C_{t_{i-1}})^{-1/2-n_i} \right]<\infty
	\end{equation}
for all $(n_1\ldots,n_d)\in\N_0^d$ with $\sum_{i=1}^d n_i\leq k\in\N_0$, for some $k\in\N_0$, then $g_{t_1,\ldots,t_d}(y)$ and $\e^{\sum_{i=1}^d\epsilon_i y_i^2 } g_{t_1,\ldots,t_d}(y)$ are uniformly bounded and $g_{t_1,\ldots,t_d}(y)$ is $k$-times continuously differentiable on $\R^d$. Property~\eqref{assint_2} implies~\eqref{assintmd} for any $k\ge 0$.
\end{theorem}

Since any uniformly bounded and integrable function on $\R^d$ is square integrable on $\R^d$, as an immediate consequence of Theorem~\ref{thmdensitymd} we have the following corollary.
\begin{corollary}\label{corL2wmd}
Assume \eqref{assintmd} holds for $k=0$. Then
	\begin{equation*}\label{E:L2md}
	\int_{\R^d} \frac{g_{t_1,\ldots,t_d}(y)^2}{\prod_{i=1}^d w_i(y_i)}\,dy <\infty
	\end{equation*}
for all Gaussian densities $w_i(y_i)$ with variances $\sigma_{w_i}^2$ satisfying
	\begin{equation}\label{siconmd}
	\sigma_{w_i}^2 >\frac{v_{max}(t_i-t_{i-1})}{2},\quad i=1,\dots,d.
	\end{equation}
\end{corollary}

\begin{remark}\label{R:finitedimdensities}
There is a one-to-one correspondence between the vector of log returns $(Y_{t_1},\ldots,Y_{t_d})$ and the vector of log prices $(X_{t_1},\ldots,X_{t_d})$. Indeed,
	\[X_{t_i}=X_0+\sum_{j=1}^i Y_{t_j}.\]
Hence, a crucial consequence of Theorem~\ref{thmdensitymd} is that the finite-dimensional distributions of the process $X_t$ admit densities with nice decay properties. More precisely, the density of $(X_{t_1},\ldots,X_{t_d})$ is $g_{t_1,\ldots,t_d}(x_1-X_0,\ldots,x_d-x_{d-1})$.
\end{remark}

Suppose that the discounted payoff of an exotic option is of the form $f(X_{t_1},...,X_{t_d})$. Assume that~\eqref{assintmd} holds with $k=0$. Set the weight function ${w(y)=\prod_{i=1}^d w_i(y_i)}$, where $w_i(y)$ is a Gaussian density with mean $\mu_{w_i}$ and variance $\sigma_{w_i}^2$ satisfying~\eqref{siconmd}. Define
	\[\widetilde{f}(y)=f(X_0+y_1,X_0+y_1+y_2,\ldots,X_0+y_1+\cdots+y_d).\]
Then by similar arguments as in Section~\ref{sec:optionprice} the price of the option is
	\[\pi_f=\E\left[f(X_{t_1},...,X_{t_d})\right]=\sum_{n_1,\ldots,n_d\geq 0}\widetilde{f}_{n_1,\ldots,n_d}\ell_{n_1,\ldots,n_d}\]
where the Fourier coefficients $\widetilde{f}_{n_1,\ldots,n_d}$ and the Hermite moments $\ell_{n_1,\ldots,n_d}$ are given by
\[\widetilde{f}_{n_1,\ldots,n_d}=\dotprod{\widetilde{f}}{H_{n_1,\ldots,n_d}}=\int_{\R^d}\widetilde{f}(y)H_{n_1,\ldots,n_d}(y)w(y)\,dy\]
and
\begin{equation}\label{E:Hmommentmd}
	\ell_{n_1,\ldots,n_d}=\E\big[H_{n_1,\ldots,n_d}(Y_{t_1},\ldots,Y_{t_d})\big]
	\end{equation}
with $H_{n_1,\ldots,n_d}(y_1,\ldots,y_d)=\prod_{i=1}^d H^{(i)}_{n_i}(y_i)$, where $H^{(i)}_{n_i}(y_i)$ is the generalized Hermite polynomial of degree $n_i$ associated to parameters $\mu_{w_i}$ and $\sigma_{w_i}$, see \eqref{HHcal}. The price approximation at truncation order $N\ge 1$ is given, in analogy to \eqref{eqproxyprice}, by
\begin{equation}\label{eqproxypriceMD}
  \pi_f^{(N)} = \sum_{n_1+\cdots+n_d=0}^N \widetilde{f}_{n_1,\ldots,n_d} \ell_{n_1,\ldots,n_d},
\end{equation}
so that $\pi_f^{(N)}\to \pi_f$ as $N\to\infty$.


We now derive universal upper and lower bounds on the implied volatility for the exotic option with discounted payoff function $f(X_{t_1},...,X_{t_d})$ and price $\pi_f$. We denote by
\begin{equation}\label{SBSdN}
 dS^{{\rm BS}}_t = S^{{\rm BS}}_t (r-\delta)\,dt + S^{{\rm BS}}_t \sigma_{\rm BS}\,dB_t
\end{equation}
the Black--Scholes price process with volatility $\sigma_{\rm BS}>0$ where $B_t$ is some Brownian motion. The Black--Scholes price is defined by
\[ \pi^{\sigma_{\rm IV}}_f =\E\Big[ f\left(\log S^{{\rm BS}}_{t_1},\dots,\log S^{{\rm BS}}_{t_d}\right) \Big] .\]
The implied volatility $\sigma_{\rm IV}$ is the volatility parameter $\sigma_{\rm BS}$ that renders the Black--Scholes option price $\pi^{\sigma_{\rm IV}}_f=\pi_f$. The following theorem provides bounds on the values that $\sigma_{\rm IV}$ may take.

\begin{theorem}\label{thm:IVEx}
Assume that the payoff function $f(\log(s_1),\dots,\log(s_d))$ is convex in the prices $(s_1,\dots,s_d)\in (0,\infty)^d$. Then the implied volatility satisfies ${\sqrt{v_{min}} \le \sigma_{\rm IV} \le  \sqrt{v_{max}}}$.
\end{theorem}

\subsection*{Examples}

We provide some examples of exotic options on the asset with price $S_t=\e^{X_t}$ for which our method applies.

The payoff of a {\textit{forward start call option on the underlying return}} between dates $t$ and $T$, and with strike $K$ is $( S_T/S_t - K )^+$ and its discounted payoff function is given by
 	\[ \widetilde{f}(y)=\e^{-r T} \left( \e^{y_2} - K\right)^+ \]
with the times $t_1=t$ and $t_2=T$. Note that $\widetilde{f}(y)= \widetilde{f}(y_2)$ only depends on $y_2$, so that this example reduces to the univariate case. In particular, the Fourier coefficients $\tilde{f}_n$ coincide with those of a call option and, as we shall see in Theroem~\ref{thmhermitemoments_md}, the \textit{forward} Hermite moments $\ell_n^\ast=\E[H_n(X_{t_2}-X_{t_1})]$ can be computed efficiently. Theorem~\ref{thm:IVEx} applies in particular to the forward start call option on the underlying return, so that its implied volatility is uniformly bounded for all maturities $T>t$. On the other hand, we know from \citet{jacquier2015asymptotics} that in the Heston model the same implied volatility explodes (except at the money) when $T\to t$.

The payoff of a \textit{forward start call option} with maturity $T$, strike fixing date $t$ and proportional strike $K$ is $(S_T - K S_t )^+$ and its discounted payoff function is given by
 	\[ \widetilde{f}(y)= \e^{-r T} \left( \e^{X_{0} + y_1 + y_2} - K\e^{X_{0} + y_1 }\right)^+ \]
with the times $t_1=t$ and $t_2=T$.
In this case the Fourier coefficients have the form
	\begin{equation*}\label{eq:fourier_fwdoptionmd}
	\begin{aligned}
	\tilde{f}_{n_1,n_2}&=\e^{X_0-rT}\int_{\R^2}\e^{y_1}H_{n_1}(y_1)w_1(y_1)(\e^{y_2}-K)^+H_{n_2}	(y_2)w_2(y_2)\,dy_1\,dy_2\\
	&=\e^{X_0-rT}f_{n_1}^{(0,-\infty)}f_{n_2}^{(0,\log K)} =f_{n_2}^{(0,\log K)}\frac{\sigma_w^{n_1}}{\sqrt{n_1!}}\e^{X_0-rT+\mu_{w_1}+\sigma_{w_1}^2/2},
	\end{aligned}
	\end{equation*}
where $f_n^{(r,k)}$ denotes the Fourier coefficient of a call option for interest rate $r$ and log strike $k$ as in~\eqref{E:fourier_coef}. Here we have used~\eqref{E:fourier_coef}--\eqref{E:recursionI} to deduce that $f_{n_1}^{(0,-\infty)}=\frac{\sigma_w^{n_1}}{\sqrt{n_1!}}\e^{\mu_{w_1}+\sigma_{w_1}^2/2}$.
In particular no numerical integration is needed. Additionally, the Hermite moments
\begin{equation*}\label{eq:hermitemommd}
\ell_{n_1,n_2}=\E\big[H_{n_1}(Y_{t_1})H_{n_2}(Y_{t_2})\big]
\end{equation*}
can be calculated efficiently as explained in Theorem~\ref{thmhermitemoments_md}. The pricing of forward start call options (on the underlying return) in the Black--Scholes model is straightforward. Analytical expressions for forward start call options (on the underlying return) have been provided in the Heston model by \citet{kruse2005pricing}. However, these integral expressions involve the Bessel
function of first kind and are therefore rather difficult to implement numerically.

The payoff of an \textit{Asian call option} with maturity $T$, discrete monitoring dates $t_{1}< \cdots <t_{d}= T$, and fixed strike $K$ is $( \sum_{i=1}^d S_{t_i}/d - K)^+$ and its discounted payoff function is given by
 	\[ \widetilde{f}(y)= \e^{-r T} \left(\frac{1}{d}\sum_{i=1}^d \e^{X_0+\sum_{j=1}^{i} y_i} - K\right)^+.\]
The payoff of an \textit{Asian call option with floating strike} is ${(S_T -  K\sum_{i=1}^d S_{t_i}/d)^+}$ and its discounted payoff function is given by
 	\[ \widetilde{f}(y)= \e^{-r T} \left( \e^{X_0+\sum_{j=1}^{d} y_j} -  \frac{K}{d}\sum_{i=1}^d \e^{X_0+\sum_{j=1}^{i} y_j} \right)^+.\]
The valuation of Asian options with continuously monitoring in the Black--Scholes model has been studied in \citet{rogers1995value} and \citet{yor2001bessel} among others.

\begin{remark}\label{rem:cuba}
The Fourier coefficients may not be available in closed-form for some exotic options, such as the Asian options. In this case, we compute the multi-dimensional version of the approximation \eqref{eqproxyprice} via numerical integration of \eqref{eqpihNint} with respect to a Gaussian density $w(x)$ in $\R^d$. This can be efficiently implemented using Gauss-Hermite quadrature, see for example \citet{jackel2005note}. Specifically, denote $z_m\in\R^d$ and $w_m\in(0,1)$ the $m$-th point and weight of an $d$-dimensional standard Gaussian cubature rule with $M$ points.
The price approximation can then be computed as follows
\begin{equation}\label{eq:quadapp}
\begin{aligned}
\pi_f^{(N)} &= \int_{\R^d} \tilde{f}\big( \mu + \Sigma z\big) \; \ell^{(N)}\big( \mu + \Sigma z\big) \;  \frac{1}{(2\pi)^{\frac d 2}}\e^{-\frac{\lVert z \rVert^2}{2}} dz \\
& \approx \sum_{m=1}^M w_m \, \tilde{f}_m \; \sum_{n_1+\dots+n_d\le N} \, \ell_{n_1,\dots,n_d} \; \prod_{i=1}^d \, \frac{1}{\sqrt{n_i !}}\Hcal_{n_i}(z_{m,i}) 
\end{aligned}
\end{equation}
where $\mu=(\mu_{w_1},\dots,\mu_{w_d})^\top$, $\Sigma = \diag(\sigma_{w_1},\dots,\sigma_{w_d})$, $\tilde{f}_m=\tilde{f}( \mu + \Sigma z_m)$, and $ \Hcal_n$ denotes the standard Hermite polynomial~\eqref{PHdef}. We emphasize that many elements in the above expression can be precomputed. A numerical example is given for the Asian option in Section~\ref{sec:exotic_num} below.
\end{remark}


\section{Numerical analysis}
\label{sec:numeric}

We analyse the performance of the price approximation \eqref{eqproxyprice} with closed-form Fourier coefficients and numerical integration of \eqref{eqpihNint} for European call options, forward start and Asian options. This includes price approximation error, model implied volatility, and computational time. The model parameters are fixed as: ${r=\delta=X_0=0}$, ${\kappa=0.5}$, ${\theta=V_0=0.04}$, ${v_{min}=10^{-4}}$, ${v_{max}=0.08}$, ${\rho=-0.5}$, and ${\sigma=1}$. 
The parameter values are in line with what could be obtained from a calibration to market prices, such as S\&P500  option prices, with the exception of $v_{max}$ that is set smaller than the typical fitted value. 
The choice $v_{max}=0.08$ permits to match the first two moments of $w(x)$ and $g(x)$ as in~\eqref{mm2}, which improves the convergence of the approximation~\eqref{eqproxyprice}. We refer to \citet{ackerer2017option} for an extension of the polynomial option pricing method, which works well for arbitrary parameter values.

\subsection{European call option}

Figure~\ref{fig:fnlnpin} displays Hermite moments $\ell_n$, Fourier coefficients $f_n$, and approximation option prices $\pi^{(N)}_f$ for a European call option with maturity $T=1/12$ and log strike $k=0$ (ATM) as functions of the truncation order $N$. The first two moments of the Gaussian density $w(x)$ match the first two moments of $X_T$, see \eqref{mm2}.\footnote{In practice, depending on the model parameters, this may not always be feasible, in which case the truncation order $N$ should be increased.} We observe that the $\ell_n$ and $f_n$ sequences oscillate and converge toward zero.
The amplitudes of these oscillations negatively impact the speed at which the approximation price sequence converges.
The gray lines surrounding the price sequence are the upper and lower price error bounds computed as in \eqref{BAERR} and Lemma~\ref{lemellnorm}, using $10^5$ Monte-Carlo samples. The price approximation converges rapidly.

Table~\ref{tab:conv} reports the implied volatility values and absolute errors in percentage points for the log strikes $k=\{-0.1,\,0,\,0.1\}$ and for various truncation orders. The reference option prices have been computed at truncation order $N=50$. For all strikes the truncation order $N=10$ is sufficient to be within 10 basis points of the reference implied volatility.

Figure~\ref{fig:smile} displays the implied volatility smile for various $v_{min}$ and $v_{max}$ such that $\sqrt{v_{min} v_{max}}=\theta$, and for the Heston model~\eqref{sdeHeston}. We observe that the smile of the Jacobi model approaches the Heston smile when $v_{min}$ is small and $v_{max}$ is large.
Somewhat surprisingly, a relatively small value for $v_{max}$ seems to be sufficient for the two smiles to coincide for options around the money.
Indeed, although the variance process has an unbounded support in the Heston model, the probability that it will visit values beyond some large threshold can be extremely small.
Figure~\ref{fig:smile} also illustrates how the implied volatility smile flattens when the variance support shrinks, $v_{max}\downarrow \theta$.
In the limit $v_{max}=\theta$, we obtain the flat implied volatility smile of the Black--Scholes model.
This shows that the Jacobi model lies between the Black--Scholes model and the Heston model and that the parameters $v_{min}$ and $v_{max}$ offer additional degrees of flexibility to model the volatility surface.

As reported in Figure~\ref{fig:cputimes}, the Fourier coefficients can be computed in less than a millisecond thanks to the recursive scheme~\eqref{E:fourier_coef}-\eqref{E:recursionI}.
Computing the Hermite moments is more costly, however they can be used to price all options with the same maturity.
The most expensive task appears to be the construction of the matrix $G$, which however is a one-off.
The Hermite moment $\ell_n$ in turn derives from the vector $v_{n,T}=\e^{G T}{\bm e}_{\pi(0,n)}$ which can be used for any initial state $(V_0,X_0)$.
Note that specific numerical methods have been developed to compute the action of the matrix exponential $\e^{G T}$ on the basis vector ${\bm e}_{\pi(0,n)}$, see for example \citet{al2011computing}, \citet{hochbruck1997krylov}, and references therein.
The running times were realized with a standard desktop computer using a single 3.5 Ghz 64 bits CPU and the \textsf{R} programming language.

\subsection{Forward start and Asian options}\label{sec:exotic_num}

The left panels of Figure~\ref{fig:fwdasian} display the approximation prices of a forward start call option with strike fixing time $t_1=1/52$ and maturity $t_2=5/52$, so that $d=2$, and of an Asian call option with weekly discrete monitoring and maturity four weeks, $t_i=i/52$ for $i\le d=4$. Both options have log strike $k=0$.
The price approximations at order $N$ have been computed using \eqref{eqproxypriceMD}. For the forward start call option, we match the first two moments of $w_i(y_i)$ and $Y_{t_i}$. For the Asian call option, we chose $\sigma_{w_i}=\sqrt{v_{max}/104}+10^{-4}$ and $\mu_{w_i}=E[X_{1/52}]$, which is in line with \eqref{siconmd} but does not match the first two moments of $Y_{t_i}$. The Fourier coefficients are not available in closed-form for the Asian call option, therefore we integrated its payoff function with respect to the density approximation using Gaussian cubature as described in Remark~\ref{rem:cuba}.
We observe that with exotic payoffs the price approximation sequence may require a larger order before stabilizing. For example, for the forward start price approximation it seems necessary to truncate beyond $N=15$ in order to obtain a accurate price approximation.

The Asian option price is approximated by~\eqref{eq:quadapp} whose computational cost depends on the number of elements in the double summation.
Therefore, in order to efficiently approximate the price, we used a truncation of the 4-dimensional product of the one-dimensional Gaussian quadrature with 20 points. 
More precisely, we selected the quadrature points having a weight larger than the $90\%$ quantile of all the weights. 
This means that, out of the $20^4$ initial points, $M=16\,000$ points were selected and their weights normalized.
Note that the $144\,000 $ removed points had a total weight of $7.2\times10^{-4}$ percent which is extremely small.
Hence, the selected points cover most of the non-negligible part of the multivariate Gaussian density support. 
An alternative approach would be to use optimal Gaussian quantizers, see \citet{pages2003optimal}.

The right panels of Figure~5 display the multi-index Hermite moments $\ell_{n_1,\dots,n_d}$ with multi-orders $n_1+\cdots+n_d=1,\dots,10$. 
Note that there are $\binom{N+d}{N}$ Hermite moments $\ell_{n_1,\dots,n_d}$ of total order $n_1+\cdots+n_d\le N$.
In practice, we observe that a significant proportion of the Hermite moments is negligible so that they may simply be set to zero if they are smaller than a certain threshold to be computed online. 
As for the quadrature points, doing so reduces the computational cost of approximating the option price.
Therefore, when approximating the Asian option price, we removed the Hermite moments having an absolute value smaller than the correspondning $10\%$ quantile.
For example, when $N=20$, this implies removing all the Hermite moments with an absolute value $|\ell_{n_1,\dots,n_d}|$ smaller than $2.35\times10^{-6}$.


\section{Conclusion}
\label{sec:conclusions}

The Jacobi model is a highly tractable and versatile stochastic volatility model. It contains the Heston stochastic volatility model as a limit case.
The moments of the finite dimensional distributions of the log prices can be calculated explicitly thanks to the polynomial property of the model.
As a result, the series approximation techniques based on the Gram--Charlier A expansions of the joint distributions of finite sequences of log returns allow us to efficiently compute prices of options whose payoff depends on the underlying asset price at finitely many time points.
Compared to the Heston model, the Jacobi model offers additional flexibility to fit a large range of Black--Scholes implied volatility surfaces.
Our numerical analysis shows that the series approximations of European call, put and digital option prices in the Jacobi model are computationally comparable to the widely used Fourier transform techniques for option pricing in the Heston model.
The truncated series of prices, whose computations do not require any numerical integration, can be implemented efficiently and reliably up to orders that guarantee accurate approximations as shown by our numerical analysis.
The pricing of forward start options, which does not involve any numerical integration, is significantly simpler and faster than the iterative numerical integration method used in the Heston model.
The minimal and maximal volatility parameters are universal bounds for Black--Scholes implied volatilities and provide additional stability to the model. In particular, Black--Scholes implied volatilities of forward start options in the Jacobi model do not experience the explosions observed in the Heston model.
Furthermore, our density approximation technique in the Jacobi model circumvents some limitations of the Fourier transform techniques in affine models and allows us to price discretely monitored Asian options.

\appendix\normalsize


\section{Hermite moments}
\label{appHm}

We apply Theorem~\ref{thmoments} to describe more explicitly how the Hermite moments $\ell_0,\dots,\ell_N$ in~\eqref{eq:Hmoments} can be efficiently computed for any fixed truncation order $N\ge 1$. We let $M=\dim\Pol_N$ and $\pi:\Ecal\rightarrow \{1,\ldots, M\}$ be an enumeration of the set of exponents
	\[\Ecal=\{(m,n): m,n\ge 0;\,m+n\le N\}.\]
The polynomials
	\begin{equation}\label{defhmn}
	h_{\pi(m,n)}(v,x) = v^m H_n(x),\quad (m,n)\in\Ecal
	\end{equation}
then form a basis of $\Pol_N$. In view of the elementary property
	\[H_n'(x)=\frac{\sqrt{n}}{\sigma_w}H_{n-1}(x),\quad n\ge 1,\]
we obtain that the $M\times M$--matrix $G$ representing $\Gcal$ on $\Pol_N$ has at most 7 nonzero elements in column $\pi(m,n)$ with $(m,n)\in\Ecal$ given by
\begin{equation*}\label{eqmatrixG}
	\begin{aligned}
	G_{\pi(m-2,n),\pi(m,n)}&=-\frac{\sigma^2 m(m-1) v_{max}v_{min}}{2(\sqrt{v_{max}}-\sqrt{v_{min}})^2},\quad m\ge 2;\\
	G_{\pi(m-1,n-1),\pi(m,n)}&=-\frac{\sigma\rho m\sqrt{n} v_{max}v_{min}}{\sigma_w(\sqrt{v_{max}}-\sqrt{v_{min}})^2}, \quad m,n\ge 1;\\
	G_{\pi(m-1,n),\pi(m,n)}&=\kappa\theta m+\frac{\sigma^2m(m-1) (v_{max}+v_{min})}{2(\sqrt{v_{max}}-\sqrt{v_{min}})^2}, \quad m\ge 1;\\
	G_{\pi(m,n-1),\pi(m,n)}&=\frac{(r-\delta)\sqrt{n}}{\sigma_w}+\frac{\sigma\rho m\sqrt{n} (v_{max}+v_{min})}{\sigma_w(\sqrt{v_{max}}-\sqrt{v_{min}})^2}, \quad n\ge 1;\\
	G_{\pi(m+1,n-2),\pi(m,n)}&=\frac{\sqrt{n(n-1)}}{2\sigma_w^2}, \quad n\ge 2;\\
	G_{\pi(m,n),\pi(m,n)}&=-\kappa m-\frac{\sigma^2 m(m-1)}{2(\sqrt{v_{max}}-\sqrt{v_{min}})^2}\\
	G_{\pi(m+1,n-1),\pi(m,n)}&=-\frac{\sqrt{n}}{2\sigma_w}-\frac{\sigma\rho m\sqrt{n}}{\sigma_w(\sqrt{v_{max}}-\sqrt{v_{min}})^2}, \quad n\ge 1.
	\end{aligned}
	\end{equation*}

Theorem~\ref{thmoments} now implies the following result.

\begin{theorem}\label{thmhermitemoments}
The coefficients $\ell_n$ are given by
	\begin{equation}\label{eqell}
	\ell_n=\begin{pmatrix}
	h_1(V_0,X_0) & \cdots & h_M(V_0,X_0)
\end{pmatrix} \,\e^{TG}\, \mathbf{e}_{\pi(0,n)},\quad 0\le n\le N,
	\end{equation}
where $\bm e_{i}$ is the $i$--th standard basis vector in $\R^{M}$.
\end{theorem}

\begin{remark}
The choice of the basis polynomials $h_{\pi(m,n)}$ in \eqref{defhmn} is convenient for our purposes because: 1) each column of the $M\times M$-matrix $G$ has at most seven nonzero entries. 2) The coefficients $\ell_n$ in the expansion of prices~\eqref{pihinf}, can be obtained directly from the action of $\e^{G T}$ on ${\bm e}_{\pi_{(0,n)}}$ as specified in~\eqref{eqell}.
In practice, it is more efficient to compute directly this action, rather than computing the matrix exponential $\e^{G T}$ and then selecting the $\pi_{(0,n)}$-column.
\end{remark}

We now extend Theorem~\ref{thmhermitemoments} to a multi-dimensional setting. The following theorem provides an efficient way to compute the multi-dimensional Hermite moments defined in~\eqref{E:Hmommentmd}. Before stating the theorem we fix some notation. Set $N=\sum_{i=1}^d n_i$ and $M=\dim\Pol_N$. Let $G^{(i)}$ be the matrix representation of the linear map $\Gcal$ restricted to $\Pol_N$ with respect to the basis, in row vector form,
\[	
h^{(i)}(v,x)=\begin{pmatrix}h^{(i)}_1(v,x) & \cdots & h^{(i)}_M(v,x)\end{pmatrix},
\]
with $h^{(i)}_{\pi(m,n)}(v,x)=v^mH^{(i)}_n(x)$ as in \eqref{defhmn} where $H_n^{(i)}$ is the generalized Hermite polynomial of degree $n$ associated to the parameters $\mu_{w_i}$ and  $\sigma_{w_i}$, see \eqref{HHcal}. Define the $M\times M$-matrix $A^{(k,l)}$ by
\[ 	A^{(k,l)}_{i,j} =
   	\begin{cases}
     	H_n^{(l)}(0) & \text{if $i=\pi(m,k)$ and $j=\pi(m,n)$ for some $m,n\in\N$} \\
     	0 & \text{otherwise.}
   	\end{cases}\]

\begin{theorem}\label{thmhermitemoments_md}
For any $n_1,\ldots,n_d\in\N_0$, the multi-dimensional Hermite moment in~\eqref{E:Hmommentmd} can be computed through
	\begin{equation*}
	\ell_{n_1,\ldots,n_d}=h^{(1)}(V_0,0) \left(\prod_{i=1}^{d-1}\e^{G^{(i)} \Delta t_i}A^{(n_{i},i+1)} \right)\e^{G^{(d)}\Delta t_d}{\bm e}_{\pi(0,n_d)},
	\end{equation*}
where $\Delta t_i=t_i-t_{i-1}$.
\end{theorem}
\begin{proof}
By an inductive argument it is sufficient to illustrate the case $n=2$. Applying the law of iterated expectation we obtain
\[
\ell_{n_1,n_2}= \E\left[H^{(1)}_{n_1}(Y_{t_1})H^{(2)}_{n_2}(Y_{t_2})\right] = \E\left[H^{(1)}_{n_1}(X_{t_1}-X_{0})\E_{t_1}\big[H^{(2)}_{n_2}(X_{t_2}-X_{t_1})\big]\right].
\]
Since the increment $X_{t_2}-X_{t_1}$ does not depend on $X_{t_1}$ we can rewrite, using Theorem~\ref{thmoments},
\[
\E_{t_1}\left[H^{(2)}_{n_2}(X_{t_2}-X_{t_1})\right] = \E\Big[H^{(2)}_{n_2}(X_{\Delta t_2})\Bigm| X_{0}=0, V_0=V_{t_1}\Big] =  h^{(2)}(V_{t_1},0)  v^{(n_2,2)}
\]
where $v^{(n_2,2)}=e^{G^{(2)} \Delta t_2}{\bm e}_{\pi(0,n_2)}$.
Note that this last expression is a polynomial solely in $V_{t_1}$
\[
h^{(2)}(V_{t_1},0)   v^{(n_2,2)} = \sum_{n=0}^{n_2} a_n \, V_{t_1}^n, \quad \text{with } a_n = \sum_{n+j\le n_2} \, H_j^{(2)}(0) \, v^{(n_2,2)}_{\pi(n,j)}.
\]
Theorem~\ref{thmoments} now implies that the Hermite coefficient is given by
\[
\ell_{n_1,n_2}= \E\big[p(V_{t_1},X_{t_1})\bigm| X_0=0\big] = h^{(1)}(V_0,0)  \e^{G^{(1)} \Delta t_1} \vec{p}\]
where $\vec{p}$ is the vector representation in the basis $h^{(1)}(v,x)$ of the polynomial
\[
p(v,x) = \sum_{n=0}^{n_2} a_n \, v^n \,  H_{n_1}(x) = h^{(1)}(v,x)  \vec{p} .
\]
We conclude by observing that the coordinates of the vector $\vec{p}$ are given by $\e_i^\top \, \vec{p} = a_n$ if $i=\pi(n, n_1)$ for some integer $n\le n_2$ and equal to zero otherwise, which in turn shows that $\vec{p} = A^{(n_1,2)} \, v^{(n_2,2)} $.
\end{proof}


\section{Proofs}
\label{a:proofs}

This appendix contains the proofs of all theorems and propositions in the main text.

\subsection*{Proof of Theorem~\ref{thmexiuni}}

For strong existence and uniqueness of~\eqref{sdeXV}, it is enough to show strong existence and uniqueness for the SDE for $V_t$,
	\begin{equation}
	\label{SDEofV}dV_t = \kappa(\theta   - V_t)\,dt + \sigma\sqrt{Q(V_t)}\,dW_{1t}.
	\end{equation}
Since the interval $[0,1]$ is an affine transformation of the unit ball in $\R$, weak existence of a $[v_{min},v_{max}]$-valued solution can be deduced from~\citet[Theorem 2.1]{LarssonPulido2015}. Path-wise uniqueness of solutions follows from~\citet[Theorem 1]{yamada1971}. Strong existence of solutions for the SDE~\eqref{SDEofV} is a consequence of path-wise uniqueness and weak existence of solutions, see for instance~\citet[Corollary 1]{yamada1971}.

Now let $v\in [v_{min},v_{max})$. The occupation times formula \citet[Corollary~VI.1.6]{revuz1999continuous} implies
\[ \int_0^\infty \1_{\{V_t=v\}} \sigma^2 Q(v)\,dt =0,\quad v>v_{min}.\]
Since $\sigma^2 Q(v)>0$ this proves~\eqref{Vtvstart} for $v>v_{min}$.
We can show that the local time at $v_{min}$ of $V_t$ is zero as in \citet[Theorem 5.3]{filipovic2015polpres} which in turn proves~\eqref{Vtvstart} for $v=v_{min}$ by applying \cite[Lemma A.1]{filipovic2015polpres}.


To conclude,~Proposition 2.2 in~\citet{LarssonPulido2015} shows that $V_t\in(v_{min},v_{max})$ if and only if $V_0\in (v_{min},v_{max})$ and condition~\eqref{eq:boundary} holds.

\subsection*{Proof of Theorem~\ref{thmconv}}

The proof of Theorem~\ref{thmconv} builds on the following four lemmas.

\begin{lemma}\label{lemweakconvmoments}
 Suppose that $Y$ and $Y^{(n)}$, $n\ge 1$, are random variables in $\R^d$ for which all moments exist. Assume further that
    \begin{equation}
        \lim_n\E\big[p(Y^{(n)})\big]=\E\big[p(Y)\big],
    \end{equation}
for any polynomial $p(y)$ and that the distribution of $Y$ is determined by its moments. Then the sequence $Y^{(n)}$ converges weakly to $Y$ as $n\to\infty$.
\end{lemma}

\begin{proof}
    Theorem 30.2 in~\citet{billingsley1995probability} proves this result for the case $d=1$. Inspection shows that the proof is still valid for the general case.
\end{proof}

\begin{lemma}\label{lemfindimmom}
The moments of the finite-dimensional distributions of the diffusions $(V_t^{(n)},X_t^{(n)})$ converge to the respective moments of the finite-dimensional distributions of $(V_t,X_t)$. That is, for any $0\le t_1<\cdots<t_d<\infty$ and for any polynomials $p_1(v,x),\ldots,p_d(v,x)$ we have
	\begin{equation}\label{eqlemmomentconv0}
	\lim_{n}\E \left[\prod_{i=1}^d p_i(V^{(n)}_{t_i},X^{(n)}_{t_i})\right]=\E \left[\prod_{i=1}^d p_i(V_{t_i},X_{t_i})\right].
	\end{equation}
\end{lemma}

\begin{proof}
Let $N=\sum_{i=1}^d \deg p_i$. Throughout  the proof we fix a basis of $\Pol_N$, $h_j(v,x)$ where $1\le j\le M=\dim\Pol_N$,  and for any polynomial $p(v,x)$ we denote by $\vv{p}$ its coordinates with respect to this basis. We denote by $G$ and $G^{(n)}$ the respective $M\times M$-matrix representations of the generators restricted to $\Pol_N$ of  $(V_t,X_t)$ and $(V^{(n)}_t,X^{(n)}_t)$, respectively. We then define recursively the polynomials $q_i(v,x)$  and $q^{(n)}_i(v,x)$ for $1\le i\le d$ by
	\begin{equation*}
	\begin{aligned}
	q_d(v,x)&=q^{(n)}_d(v,x)=p_d(v,x),\\
	q_i(v,x)&= p_i(v,x) \begin{pmatrix}h_1(v,x) & \cdots & h_{M}(v,x)\end{pmatrix}\e^{(t_{i+1}-t_i)G}\vv{q_{i+1}} ,\quad 1\le i<d,\\
	q^{(n)}_{i}(v,x)&= p_i(v,x)  \begin{pmatrix}h_1(v,x) & \cdots & h_{M}(v,x)\end{pmatrix}\e^{(t_{i+1}-t_i)G^{(n)}}\vv{q^{(n)}_{i+1}},\quad 1\le i<d.\\
	\end{aligned}	
	\end{equation*}
As in the proof of Theorem~\ref{thmhermitemoments_md}, a successive application of Theorem~\ref{thmoments} and the law of iterated expectation implies that
	\begin{align*}
	\E \left[\prod_{i=1}^d p_i(V_{t_i},X_{t_i})\right]&=\E \left[\prod_{i=1}^{d-1} p_i(V_{t_i},X_{t_i}) \E\big[ p_d(V_{t_d},X_{t_d}) \bigm| \Fcal_{t_{d-1}}\big]\right]\\
	&=\cdots=\begin{pmatrix}h_1(V_0,X_0) & \cdots & h_{M}(V_0,X_0)\end{pmatrix}\e^{t_1G}\vv{q_1}.
	\end{align*}
and similarly,
	\[ 
	\E \left[\prod_{i=1}^d p_i(V^{(n)}_{t_i},X^{(n)}_{t_i})\right]=\begin{pmatrix}h_1(V_0^{(n)},X_0^{(n)}) & \cdots & h_{M}(V_0^{(n)},X_0^{(n)})\end{pmatrix}\e^{t_1G^{(n)}}\vv{q_1^{(n)}}.
	\]

We deduce from~\eqref{eqgenerator} that
	\begin{equation}\label{eqlemmomentconv1}
	\lim_n G^{(n)}=G.
	\end{equation}
This is valid also for the limit case $v_{max}=\infty$, that is ${Q(v)=v-v_{min}}$. This fact together with an inductive argument shows that $\lim_n\vv{q_1^{(n)}}=\vv{q_1}$. This combined with~\eqref{eqlemmomentconv1} proves~\eqref{eqlemmomentconv0}.
\end{proof}

\begin{lemma}\label{lemfindimdet}
The finite-dimensional distributions of $(V_t,X_t)$ are determined by their moments.
\end{lemma}

\begin{proof}
The proof of this result is contained in the proof of~\citet[Lemma 4.1] {filipovic2015polpres}.
\end{proof}

\begin{lemma}\label{lemtight}
The family of diffusions $(V_t^{(n)},X_t^{(n)})$ is tight.
\end{lemma}

\begin{proof}
Fix a time horizon $N\in\N$. We first observe that by~\citet[Problem V.3.15]{karatzas1991brownian} there is a constant $K$ independent of $n$ such that
	\begin{equation}\label{Kolmd}
	\E\left[\|(V_t^{(n)},X_t^{(n)})-(V_s^{(n)},X_s^{(n)})\|^4\right]\leq K|t-s|^2,\quad 0\le s<t\le N.
	\end{equation}
Now fix any positive $\alpha<1/4$. Kolmogorov's continuity theorem (see \citet[Theorem I.2.1]{revuz1999continuous}) implies that
	\[ \E\left[\left(\sup_{0\le s<t\le N} \frac{\|(V_t^{(n)},X_t^{(n)})-(V_s^{(n)},X_s^{(n)})\|}{|t-s|^\alpha}\right)^4\right]\le J \]
for a finite constant $J$ that is independent of $n$. The modulus of continuity
	\[ \Delta(\delta,n)=\sup\Big\{ \|(V_t^{(n)},X_t^{(n)})-(V_s^{(n)},X_s^{(n)})\| : 0\le s<t\le N,\, |t-s|<\delta\Big\}\]
thus satisfies
	\[ \E\left[ \Delta(\delta,n)^4\right] \le \delta^\alpha J .\]
Using Chebyshev's inequality we conclude that, for every $\epsilon>0$,
	\[ \Q\left[ \Delta(\delta,n)>\epsilon\right] \le \frac{\E[ \Delta(\delta,n)^4]}{\epsilon^4}\le \frac{  \delta^\alpha J}{\epsilon^4}, \]
and thus $\sup_{n} \Q[ \Delta(\delta,n)>\epsilon]\to 0$ as $\delta\to 0$. This together with the property that the initial states $(V_0^{(n)},X_0^{(n)})$ converge to $(V_0,X_0)$ as $n\to\infty$ proves the lemma, see~\citet[Theorem II.85.3]{rogers2000diffusions}.\footnote{The derivation of the tightness of $(V_t^{(n)},X_t^{(n)})$ from \eqref{Kolmd} is also stated without proof in \citet[Theorem II.85.5]{rogers2000diffusions}. For the sake of completeness we give a short self-contained argument here.}
\end{proof}

\begin{remark}
Kolmogorov's continuity theorem (see \citet[Theorem I.2.1]{revuz1999continuous}) and \eqref{Kolmd} imply that the paths of $(V_t,X_t)$ are $\alpha$-H\"older continuous for any $\alpha<1/4$.
\end{remark}

Lemmas~\ref{lemweakconvmoments}--\ref{lemfindimdet} imply that the finite-dimensional distributions of the diffusions $(V_t^{(n)},X_t^{(n)})$ converge weakly to those of $(V_t,X_t)$ as $n\to\infty$. Theorem~\ref{thmconv} thus follows from Lemma~\ref{lemtight} and \citet[Lemma II.87.3]{rogers2000diffusions}.

\subsection*{Proof of Theorem~\ref{thmoptionFC}}

We claim that the solution of the recursion~\eqref{E:recursionI} is given by
	\begin{equation}\label{E:functionI}
	I_n(\mu;\nu)=\int_\mu^{\infty}\Hcal_n(x)\e^{\nu x}\phi(x)\,dx,\quad n\ge 0.
	\end{equation}
Indeed, for $n=0$ the right hand side of \eqref{E:functionI} equals
	\[ \int_\mu^{\infty}\Hcal_0(x)\e^{\nu x}\phi(x)\,dx=\e^{\frac{\nu^2}{2}}\int_{\mu-\nu}^{\infty}\phi(x)\,dx,\]
which is $I_0(\mu;\nu)$. For $n\ge 1$, we recall that the standard Hermite polynomials $\Hcal_{n}(x)$ satisfy
	\begin{equation}\label{E:recursion}
	\Hcal_{n}(x)= x\Hcal_{n-1}(x)-\Hcal_{n-1}'(x).
	\end{equation}
Integration by parts and~\eqref{E:recursion} then show that
	\begin{align*}
		\int_\mu^{\infty}\Hcal_n(x)\e^{\nu x}\phi (x)\,dx&= \int_{\mu}^{\infty}\Hcal_{n-1}(x)\e^{\nu x}x\phi (x)\,dx-\int_{\mu}^{\infty}\Hcal_{n-1}'(x)\e^{\nu x}\phi (x)\,dx \\
		&=- \Hcal_{n-1}(x)\e^{\nu x}\phi (x)\big|_{\mu}^{\infty}+ \int_{\mu}^{\infty}\Hcal_{n-1}(x)\nu \e^{\nu x}\phi (x)\,dx.\\
		&= \Hcal_{n-1}(\mu)\e^{\nu \mu}\phi (\mu)+\nu\int_{\mu}^{\infty}\Hcal_{n-1}(x)\e^{\nu x}\phi (x)\,dx ,	
	\end{align*}
which proves \eqref{E:functionI}.

A change of variables, using \eqref{HHcal} and~\eqref{E:functionI}, shows
	\begin{align*}
	f_n &= \e^{-rT} \int_k^\infty \left( \e^x -\e^k\right) H_n(x) w(x)\,dx\\
	&= \e^{-rT} \int_{\frac{k-\mu_w}{\sigma_w}}^\infty \left( \e^{\mu_w+\sigma_w z} -\e^k\right) H_n(\mu_w+\sigma_w z) w(\mu_w+\sigma_w z)\sigma_w\,dz\\
	&= \e^{-rT} \frac{1}{\sqrt{n!}}\int_{\frac{k-\mu_w}{\sigma_w}}^\infty \left( \e^{\mu_w+\sigma_w z} -\e^k\right) \Hcal_n(z) \phi(z)\,dz\\
	&=\e^{-rT+\mu_w}\frac{1}{\sqrt{n!}}I_n\left(\frac{k-\mu_w}{\sigma_w};\sigma_w\right)-\e^{-rT+k}\frac{1}{\sqrt{n!}}I_n\left(\frac{k-\mu_w}{\sigma_w};0\right).
	\end{align*}
Formulas~\eqref{E:fourier_coef} follow from the recursion formula \eqref{E:recursionI}.

\subsection*{Proof of Theorem~\ref{thmoptionDigital}}

As before, a change of variables, using \eqref{HHcal} and~\eqref{E:functionI}, shows
	\begin{align*}
	f_n&=\e^{-rT} \int_k^\infty H_n(x) w(x)\,dx = \frac{\e^{-rT}}{\sqrt{n!}} \int_{\frac{k-\mu_w}{\sigma_w}}^\infty  \Hcal_n(z) \phi(z)\,dz \\&=\frac{\e^{-rT}}{\sqrt{n!}} I_n\left(\frac{k-\mu_w}{\sigma_w};0\right).
	\end{align*}
Formulas~\eqref{E:fourier_coef_dig} follow directly from \eqref{E:recursionI}.

\subsection*{Proof of Lemma~\ref{lemellnorm}}

 We use similar notation as in the proof of Theorem~\ref{thmdensitymd}. In particular, with $C_T$ as in~\eqref{eq:defC_T} and $M_T$ as in~\eqref{eq:defM_Tnew}, we denote by
 \begin{equation}\label{GxdefN}
  G_T(x)=(2\pi C_T)^{-\frac12}\exp\left(-\frac{(x-M_T)^2}{2C_T}\right)
 \end{equation}
the conditional density of $X_T$ given $\{V_t:t\in[0,T]\}$, so that $g_T(x)=\E[G_T(x)]$ is the unconditional density of $X_T$. Lemma~\ref{lemellnorm} now follows from observing that $G_T(x)=\phi(x,M_T,C_T)$ and $w(x) = \phi(x,\mu_w,\sigma_w^2)$.

\subsection*{Proof of Lemma~\ref{lemdecayH}}

We first recall that by Cram\'er's inequality (see for instance~\citet[Section 10.18]{erd_53}) there exists a constant $K>0$ such that for all $n\ge 0$
	\begin{equation}\label{cramer}
	\e^{-(x-\mu_w)^2/4\sigma_w^2}|H_n(x)|=(n!)^{-1/2}\e^{-(x-\mu_w)^2/4\sigma_w^2}\left|\Hcal_n\left(\frac{x-\mu_w}{\sigma_w}\right)\right|\leq K.
	\end{equation}
Additionally, as in the proof Theorem~\ref{thmdensitymd}, since $1/4\sigma^2_w<1/(2v_{max}T)$,
	\[ \E\left[\int_{\R}\e^{(x-\mu_w)^2/4\sigma_w^2}G_T(x)\,dx\right]<\infty,\]
where $G_T(x)$ is given in \eqref{GxdefN}. This implies
	\begin{equation*}
	\begin{aligned}
	\E&\left[\int_{R}|H_n(x)|G_T(x)\,dx\right]\\
	&\quad=\E\left[\int_{R}|H_n(x)|\e^{-(x-\mu_w)^2/4\sigma_w^2}\e^{(x-\mu_w)^2/4\sigma_w^2}G_T(x)\,dx\right]\\
	&\quad\leq K \E\left[ \int_R\e^{(x-\mu_w)^2/4\sigma_w^2}G_T(x)\, dx\right]<\infty.
	\end{aligned}
	\end{equation*}
We can therefore use Fubini's theorem to deduce
	\begin{equation}\label{eq:lnexplicit}
	\ell_n=\int_{\R}H_n(x)g_T(x)\,dx
	=\E\left[\int_{\R}H_n(x)G_T(x)\,dx\right]=\E[Y_n].
	\end{equation}

We now analyze the term inside the expectation in~\eqref{eq:lnexplicit}. A change of variables shows
\[
	Y_n=\int_{R}H_n(x)G_T(x)\,dx=(2\pi n!)^{-1/2}\int_{\R}\Hcal_n(\alpha y+\beta)\e^{-y^2/2}\,dy,
\]
where we define $\alpha=\frac{\sqrt{C_T}}{\sigma_w}$ and $\beta=\frac{M_T-\mu_w}{\sigma_w}$. We recall that
\begin{equation}\label{E:boundsC}
0<(1-\rho^2)v_{min}T\leq C_T\leq v_{max}T<\sigma_w.
\end{equation}
The inequalities in~\eqref{E:boundsC} together with the fact that $V_t$ is a bounded process yield the following uniform bounds for $\alpha,\beta$,
\begin{equation}\label{E:alphabounds}
1-q=\frac{(1-\rho^2)v_{min}T}{\sigma_w^2}\leq \alpha^2\leq v_{max}T/\sigma_w^2<1,\quad|\beta|\leq R,
\end{equation}
with constants $0<q<1$ and $R>0$. Define
	\[x_n=(2\pi)^{-1/2}\int_{\R}\Hcal_n(\alpha y+\beta)\e^{-y^2/2}\,dy,\]
so that
\[	
Y_n=\int_{R}H_n(x)G_T(x)\,dx=(n!)^{-1/2}x_n.
\]
An integration by parts argument using~\eqref{E:recursion} and the identity
	\[\Hcal'_n(x)=n\Hcal_{n-1}(x)\]
shows the following recursion formula
	\[x_n=\beta x_{n-1}-(n-1)(1-\alpha^2)x_{n-2},\]
with $x_0=1$ and $x_1=\beta$. This recursion formula is closely related to the recursion formula of the Hermite polynomials which helps us deduce the following explicit expression
	\begin{equation}\label{eq:lnexplicit3}
	x_n=n!\sum_{m=0}^{\lfloor n/2\rfloor}\frac{(\alpha^2-1)^m}{m!(n-2m)!}\frac{\beta^{n-2m}}{2^m}.
	\end{equation}
Recall that
	\begin{equation}\label{eq:lnexplicit3.5}
	\Hcal_n(x)=n!\sum_{m=0}^{\lfloor n/2\rfloor}\frac{(-1)^m}{m!(n-2m)!}\frac{x^{n-2m}}{2^m}.
	\end{equation}
By~\eqref{eq:lnexplicit3} and~\eqref{eq:lnexplicit3.5} we have
\begin{align*}
x_n & =n!(1-\alpha^2)^{\frac n2}\sum_{m=0}^{\lfloor n/2\rfloor}\frac{(-1)^m}{m!(n-2m)!}\frac{((1-\alpha^2)^{-\frac12}\beta)^{n-2m}}{2^m} \\
& = (1-\alpha^2)^{\frac n2}\Hcal_n\left((1-\alpha^2)^{-\frac12}\beta\right)
\end{align*}
and
\[
\ell_n=\E\left[(1-\alpha^2)^{\frac n2}n!^{-\frac12}\Hcal_n\left((1-\alpha^2)^{-\frac12}\beta\right)\right].
\]
Cauchy-Schwarz inequality and~\eqref{cramer} yield
	\begin{equation}\label{E:estimateelln}
	\begin{split}
	\ell_n^2 & \leq \E\left[\left(n!^{-\frac 12}\Hcal_n\big((1-\alpha^2)^{-\frac12}\beta\big)\right)^2\right]\E\left[(1-\alpha^2)^n\right] \\
	& \leq K^2\E\left[\exp\Big(\beta^2/\big(2(1-\alpha^2)\big)\Big)\right]\E\left[(1-\alpha^2)^n\right].
	\end{split}
	\end{equation}
Inequalities~\eqref{E:alphabounds} and~\eqref{E:estimateelln} imply the existence of constants $C>0$ and $0<q<1$ such that
 $\ell_n^2\leq Cq^n$.

\subsection*{Proof of Theorem~\ref{thmdensitymd}}

In order to shorten the notation we write $\Delta Z_{t_i}=Z_{t_i}-Z_{t_{i-1}}$ for any process $Z_t$. From \eqref{sdeXV} we infer that the log price $X_t=  M_t + \int_0^t \sqrt{V_s -\rho^2 Q(V_s)}\,dW_{2s}$ where $M_t$ is defined in~\eqref{eq:defM_Tnew}. In particular the log returns $Y_{t_i}=\Delta X_{t_i}$ have the form
	\[Y_{t_i}=\Delta M_{t_i}+ \int_{t_{i-1}}^{t_i} \sqrt{V_s -\rho^2 Q(V_s)}\,dW_{2s}.\]
In view of property~\eqref{Vtvstart} we infer that $\Delta C_{t_i}>0$ for $i=1,\ldots,d$. Motivated by~\citet{BroadieKaya2006}, we notice that, conditional on $\{ V_{t},\, t\in [0,T]\}$, the random variable $(Y_{t_1},\ldots,Y_{t_d})$ is Gaussian with mean vector $(\Delta M_{t_1},\ldots,\Delta M_{t_d})$ and covariance matrix ${\rm diag}(\Delta C_{t_1},\ldots,\Delta C_{t_d})$. Its density  $G_{t_1,\ldots,t_d}(y)$ has the form
	\[G_{t_1,\ldots,t_d}(y)  =(2\pi)^{-d/2}\prod_{i=1}^d (\Delta C_{t_i})^{-1/2}\exp\left[-\sum_{i=1}^d \frac{ (y_i-\Delta M_{t_i})^2}{2\Delta C_{t_i}}\right].\]
Fubini's theorem implies that $g_{t_1,\ldots,t_d}(y)  = \E[ G_{t_1,\ldots,t_d}(y)]$ is measurable and satisfies, for any bounded measurable function $f(y)$,
	\[\E\left[ f(Y_{t_1},\ldots,Y_{t_d})\right] =\E\left[ \int_{\R^d} f(y)G_{t_1,\ldots,t_d}(y) \,dy\right] = \int_{\R^d} f(y)g_{t_1,\ldots,t_d}(y) \,dy.\]
Hence the distribution of $(Y_{t_1},\ldots,Y_{t_d})$ admits the density $g_{t_1,\ldots,t_d}(y)$ on $\R^d$. Dominated convergence implies that $g_{t_1,\ldots,t_d}(y)$ is uniformly bounded and $k$--times continuously differentiable on $\R^d$ if \eqref{assintmd} holds. The arguments so far do not depended on $\epsilon_i$ and also apply to the Heston model, which proves Remark~\ref{remHESTON}.

For the rest of the proof we assume, without loss of generality, that $\epsilon_i>0$ for $i=1,\ldots,d$. Observe that the mean vector and covariance matrix of $G_{t_1,\ldots,t_d}(y)$ admit the uniform bounds
	\[ |\Delta M_{t_i}| \le  K,\quad {|\Delta C_{t_i}|}\le v_{max}(t_i-t_{i-1}),\]
for some finite constant $K$. Define $\Delta_i=1-2\epsilon_i {\Delta C_{t_i}}$ and $\delta_i=1-2\epsilon_i v_{max}(t_i-t_{i-1})$. Then $\delta_i\in (0,1)$ and $\Delta_i\geq \delta_i$. Completing the square implies
	\begin{align}
	&\e^{\sum_{i=1}^d \epsilon_i y_i^2 } G_{t_1,\ldots,t_d}(y) =\prod_{i=1}^d  (2\pi\Delta C_{t_i})^{-\frac12}\exp\left[\epsilon_i y_i^2-\frac{ (y_i-\Delta M_{t_i})^2}{2{\Delta C_{t_i}}}\right]\notag \\
 	& \quad =\prod_{i=1}^d  (2\pi\Delta C_{t_i})^{-\frac12} \exp\left[ - \frac{\Delta_i}{2{\Delta C_{t_i}}}\left( y_i-\frac{\Delta M_{t_i}}{\Delta_i}\right)^2 + \frac{\Delta M_{t_i}^2}{2\Delta C_{t_i}}\left(\frac{1}{\Delta_i}-1\right)\right] \notag\\
  	& \quad =\prod_{i=1}^d   (2\pi\Delta C_{t_i})^{-\frac12} \exp\left[ - \frac{\Delta_i}{2\Delta C_{t_i}}\left( y_i-\frac{\Delta M_{t_i}}{\Delta_i}\right)^2 + \frac{\epsilon_i \Delta M_{t_i}^2}{\Delta_i}\right].\label{ex2GTx}
	\end{align}
Integration of \eqref{ex2GTx} then gives
	\[  \int_{\R^d} \e^{\sum_{i=1}^d \epsilon_i y_i^2 } G_{t_1,\ldots,t_d}(y) \,dy = \prod_{i=1}^d  \frac{1}{\sqrt{\Delta_i}} \exp\left[\frac{\epsilon_i \Delta M_{t_i}^2}{\Delta_i}\right]\leq \prod_{i=1}^d \frac{1}{\sqrt{\delta_i }} \exp\left[\frac{\epsilon_i K^2}{\delta_i }\right] .\]
Hence \eqref{gint} follows by Fubini's theorem after taking expectation on both sides. We also derive from \eqref{ex2GTx} that
\begin{align*}
e^{\sum_{i=1}^d \epsilon_i y_i^2 } g_{t_1,\ldots,t_d}(y)&=\E\left[\e^{\sum_{i=1}^d \epsilon_i y_i^2 } G_{t_1,\ldots,t_d}(y)\right]\\
&\leq \E\left[ \prod_{i=1}^d  (2\pi\Delta C_{t_i})^{-\frac12}\right]\prod_{i=1}^d \exp\left[\frac{\epsilon_i K^2}{\delta_i }\right].
\end{align*}
Hence $\e^{\sum_{i=1}^d \epsilon_i y_i^2 } g_{t_1,\ldots,t_d}(y)$ is uniformly bounded and continuous on $\R^d$ if \eqref{assintmd} holds. In fact, for this to hold it is enough suppose that \eqref{assintmd} holds with $k=0$. Moreover, \eqref{assint_2} implies that $\Delta C_{t_i}\ge (t_i-t_{i-1})(1-\rho^2) v_{min}>0$ and~\eqref{assintmd} follows.

\subsection*{Proof of Theorem~\ref{thm:IVEx}}

We assume the Brownian motions $B_t$ and $(W_{1t},W_{2t})$ in \eqref{SBSdN} and \eqref{sdeXV} are independent. We denote by $\pi_{f,t}$ the time-$t$ price of the exotic option in the Jacobi model.

For any $t_{i-1}\le t<t_i$ and given a realization $X_{t_1},\dots,X_{t_{i-1}}$, the time-$t$ Black--Scholes price of the option is a function $\pi^{\sigma_{\rm BS}}_f(t,S_t)$ of $t$ and the spot price $S_t$ defined by
\begin{align*}
   \e^{-rt} \pi^{\sigma_{\rm BS}}_f(t,s)&=\E\big[ f\left(X_{t_1},\dots,X_{t_{i-1}},\log S^{{\rm BS}}_{t_i},\dots,\log S^{{\rm BS}}_{t_d}\right) \bigm| \Fcal_t,\, S^{{\rm BS}}_t=s\big] \\
   &= \E\Big[ f\Big(X_{t_1},\dots,X_{t_{i-1}},\log \left( s R^{\rm BS}_{t,t_i}\right),\dots,\log \left( s R^{\rm BS}_{t,t_d}\right)\Big) \Bigm| \Fcal_t\Big]
  \end{align*}
where we write
\[ R^{\rm BS}_{t,t_i}=\e^{\left(r-\delta-\frac{1}{2}\sigma_{\rm BS}^2\right)(t_i-t)+\sigma_{\rm BS}\left(B_{t_i}-B_t\right)}.\]
By assumption, we infer that $\pi^{\sigma_{\rm BS}}_f(t,s)$ is convex in $s>0$. Moreover, $\pi_{f}^{\sigma_{\rm BS}}(t,s)$ satisfies the following PDE
\begin{equation}\label{BSPDE}
r \pi_{f}^{{\sigma_{\rm BS}}}(t,s) = \frac{\partial \pi_{f}^{{\sigma_{\rm BS}}}(t,s)}{\partial t} + (r-\delta) s \frac{\partial \pi_{f}^{{\sigma_{\rm BS}}}(t,s)}{\partial s} + \frac{1}{2} \sigma_{\rm BS}^2 s^2\frac{\partial^2 \pi_{f}^{{\sigma_{\rm BS}}}(t,s)}{\partial s^2}
\end{equation}
and has terminal value satisfying $\pi_{f}^{{\sigma_{\rm BS}}}(T,S_T)=\pi_{f,T}$. Write
\begin{align*}
& \pi_{f,t}^{\sigma_{\rm BS}} = \pi^{\sigma_{\rm BS}}_f(t,S_t),\quad \Theta^{\sigma_{\rm BS}}_{f,t}= - \frac{\partial \pi_{f}^{{\sigma_{\rm BS}}}(t,S_t)}{\partial t}, \\
& \Delta_{f,t}^{\sigma_{\rm BS}} =  \frac{\partial \pi_{f}^{{\sigma_{\rm BS}}}(t,S_t)}{\partial s},\quad \Gamma^{\sigma_{\rm BS}}_{f,t}=  \frac{\partial^2 \pi_{f}^{{\sigma_{\rm BS}}}(t,S_t)}{\partial s^2} 
\end{align*}
and $dN_t =\rho\, \sqrt{Q(V_t)}\,dW_{1t}  + \sqrt{V_t -\rho^2 \,Q(V_t)}\,dW_{2t} $ for the martingale driving the asset return in \eqref{sdeXV} such that, using \eqref{BSPDE},
\begin{align*}
 d(e^{-rt}\pi_{f,t}^{{\sigma_{\rm BS}}}) &= \e^{-rt} \left( -r \pi_{f,t}^{{\sigma_{\rm BS}}}  -\Theta^{\sigma_{\rm BS}}_{f,t}  +(r-\delta) S_t \Delta^{\sigma_{\rm BS}}_{f,t} +\frac{1}{2} V_t S_t^2 \Gamma^{\sigma_{\rm BS}}_{f,t}\right)dt \\
 & \quad + \e^{-rt} \Delta_{f,t}^{{\sigma_{\rm BS}}} S_t \,dN_t\\
&= \frac{1}{2} \e^{-rt} (V_t-\sigma_{\rm BS}^2) S_t^2 \Gamma^{\sigma_{\rm BS}}_{f,t}\,dt + \e^{-rt}\Delta_{f,t}^{{\sigma_{\rm BS}}} S_t \,dN_t.
\end{align*}

Consider the self-financing portfolio with zero initial value, long one unit of the exotic option, and short $\Delta_{f,t}^{{\sigma_{\rm BS}}} $ units of the underlying asset. Let $\Pi_t$ denote the time-$t$ value of this portfolio. Its discounted price dynamics then satisfies

\begin{align*}
d(e^{-rt}\Pi_t) &= d(e^{-rt}\pi_{f,t}) - \Delta_{f,t}^{{\sigma_{\rm BS}}} \left( d(e^{-rt}S_t) + \e^{-rt}S_t\delta\,dt\right)\\
&= d(e^{-rt}\pi_{f,t}) - \Delta_{f,t}^{{\sigma_{\rm BS}}}e^{-rt}  S_t\,dN_t \\
&= d(e^{-rt}\pi_{f,t}) - d(e^{-rt}\pi_{f,t}^{{\sigma_{\rm BS}}}) + \frac{1}{2}e^{-rt}(V_t-\sigma_{\rm BS}^2)S_t^2\Gamma_{f,t}^{{\sigma_{\rm BS}}}\,dt.
\end{align*}
Integrating in $t$ gives
\begin{equation}\label{E:BShedging}
\e^{-rT} \Pi_T =  -\pi_{f,0} + \pi_{f,0}^{{\sigma_{\rm BS}}} + \frac{1}{2}\int_0^Te^{-rt}(V_t-\sigma_{\rm BS}^2)S_t^2\Gamma_{f,t}^{{\sigma_{\rm BS}}}\,dt
\end{equation}
as $\pi_{f,T} - \pi_{f,T}^{{\sigma_{\rm BS}}}=0$.

We now claim that the time-$0$ option price $\pi_{f,0}=\pi_f$ lies between the Black--Scholes option prices for $\sigma_{\rm BS}=\sqrt{v_{min}}$ and $\sigma_{\rm BS}=\sqrt{v_{max}}$,
\begin{equation}\label{claimD}
\pi^{\sqrt{v_{min}}}_{f,0} \le \pi_{f} \le \pi^{\sqrt{v_{max}}}_{f,0}.
\end{equation}
Indeed, let $\sigma_{\rm BS}=\sqrt{v_{min}}$. Because $\Gamma_{f,t}^{{\rm BS}}\ge 0$ by assumption, it follows from \eqref{E:BShedging} that $\e^{-rT} \Pi_T\ge  - \pi_{f,0}+\pi_{f,0}^{\sqrt{v_{min}}}$. Absence of arbitrage implies that $\Pi_T$ must not be bounded away from zero, hence $ - \pi_{f,0}+\pi_{f,0}^{\sqrt{v_{min}}}\le 0$. This proves the left inequality in \eqref{claimD}. The right inequality follows similarly, whence the claim~\eqref{claimD} is proved.

A similar argument shows that the Black--Scholes price $\pi^{\sigma_{\rm BS}}_{f,0}$ is non-decreasing in $\sigma_{{\rm BS}}$, whence $\sqrt{v_{min}}\le \sigma_{{\rm IV}}\le \sqrt{v_{max}}$, and the theorem is proved.


\begin{table}
\centering
\begin{tabular}{l||rr|rr|rr}
& \multicolumn{2}{c|}{$k=-0.1$} & \multicolumn{2}{c|}{$k=0$} & \multicolumn{2}{c}{$k=0.1$}\\
$N$ & IV & error &  IV & error & IV & error \\
  \hline  \hline
0--2 & 20.13 & 2.62 & 20.09 & 0.86 & 20.08 & 0.83 \\
  3 & 22.12 & 0.63 & 19.96 & 0.73 & 16.60 & 2.65 \\
  4 & 23.02 & 0.27 & 19.27 & 0.04 & 18.88 & 0.37 \\
  5 & 23.03 & 0.28 & 19.27 & 0.04 & 18.88 & 0.37 \\
  6 & 22.93 & 0.18 & 19.33 & 0.10 & 18.72 & 0.53 \\
  7 & 22.76 & 0.01 & 19.32 & 0.09 & 19.11 & 0.14 \\
  8 & 22.83 & 0.08 & 19.22 & 0.01 & 19.18 & 0.07 \\
  9 & 22.82 & 0.07 & 19.22 & 0.01 & 19.19 & 0.06 \\
  10 & 22.83 & 0.08 & 19.25 & 0.02 & 19.22 & 0.03 \\
  15 & 22.74 & 0.01 & 19.23 & 0.00 & 19.32 & 0.07 \\
  20 & 22.75 & 0.00 & 19.23 & 0.00 & 19.28 & 0.03 \\
  30 & 22.75 & 0.00 & 19.23 & 0.00 & 19.25 & 0.00 \\
\end{tabular}
\caption{Implied volatility values and absolute errors in percentage points for European call option price approximations at various truncation orders $N$ and log strikes $k$.
}\label{tab:conv}
\end{table}

\begin{figure}
\centering
\includegraphics{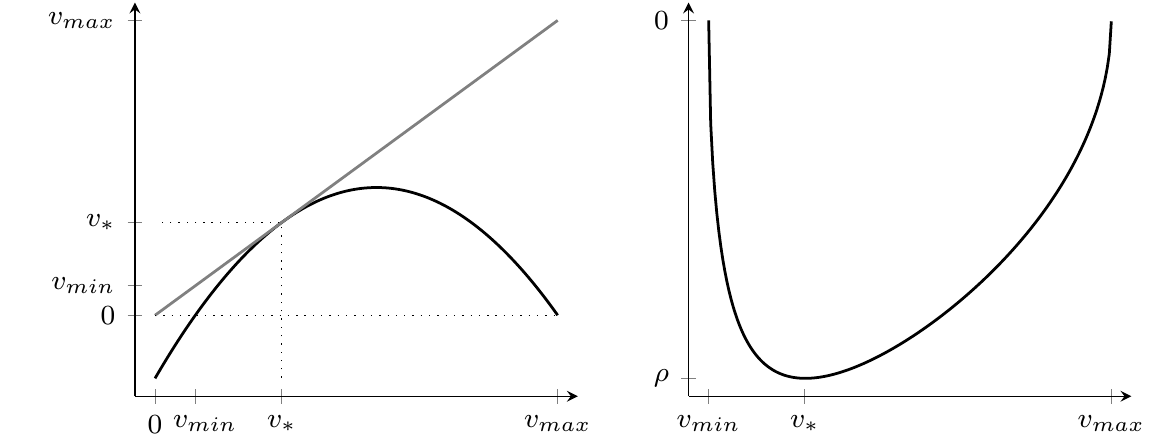}
\caption{Variance and correlation.\\
The quadratic variation of the Jacobi model (black line) and of the Heston model (gray line) are displayed in the left panel as a function of the instantaneous variance.
The right panel displays the instantaneous correlation between the processes $X_t$ and $V_t$ as a function of the instantaneous variance.
We denote $v_* = \sqrt{v_{min}v_{max}}$ and assumed that $\rho<0$.}\label{fig:varcor}
\end{figure}

\begin{figure}
\centering
\includegraphics{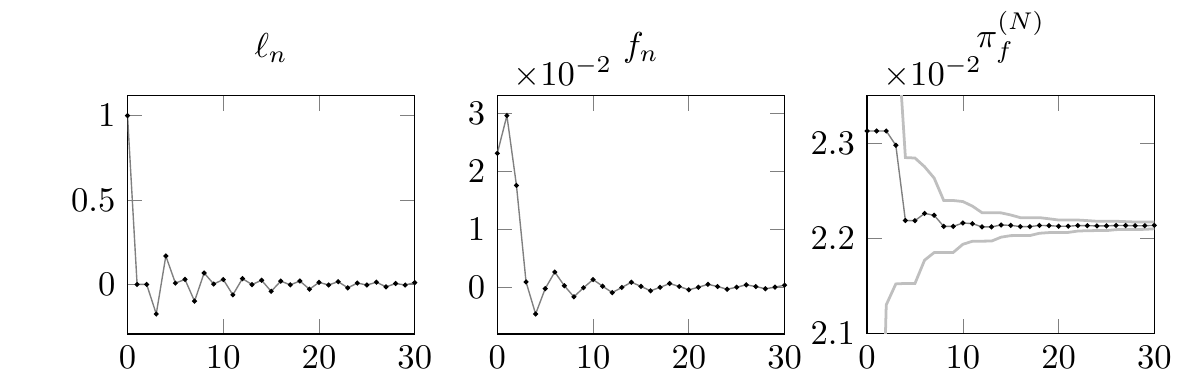}
\caption{European call option.\\
Hermite moments $\ell_n$, Fourier coefficients $f_n$, and approximation prices $\pi_f^{(N)}$ with error bounds as functions of the order $n$ (truncation order $N$).}\label{fig:fnlnpin}
\end{figure}

\begin{figure}
\centering
\includegraphics{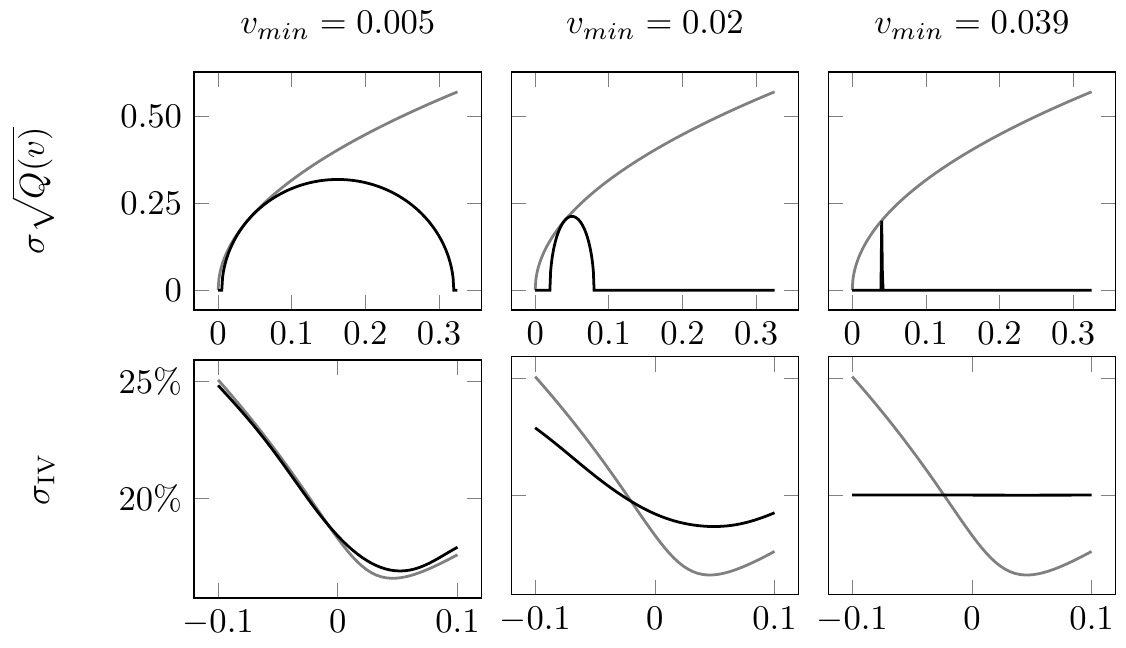}
\caption{Implied volatility smile: from Heston to Black--Scholes.\\
The first row displays the variance process' diffusion function in the Jacobi model (black line) and in the Heston model (gray line).
The second row displays the implied volatility as a function of the log strike $k$ in the Jacobi model (black line) and in the Heston model (gray line).
}\label{fig:smile}
\end{figure}

\begin{figure}
\centering
\includegraphics{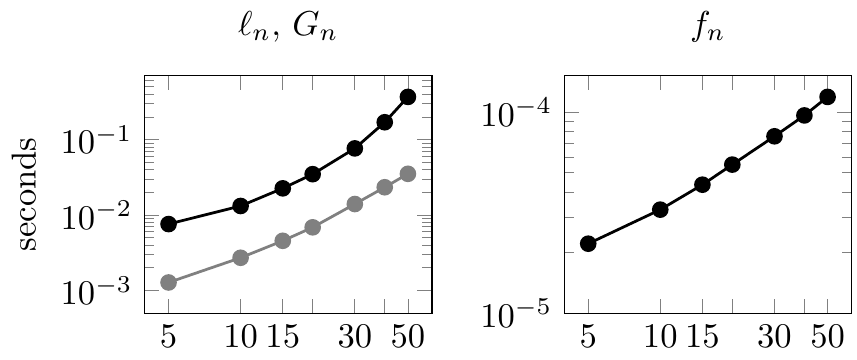}
\caption{Computational performance.\\
The left panel displays the computing time to derive the Hermite moments $\ell_n$ (black line) and the matrix $G$ (gray line) as functions of the order $n$.
The right panel displays the same relation for the Fourier coefficients $f_n$ (black line).
}\label{fig:cputimes}
\end{figure}

\begin{figure}
\centering
\includegraphics{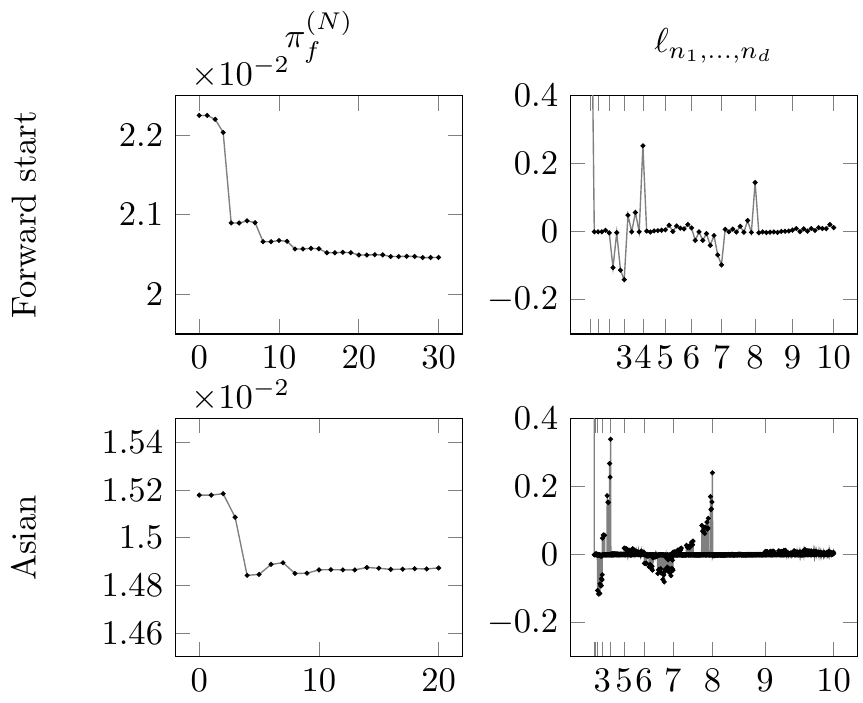}
\caption{Forward start and Asian options.\\
The left panels display the approximation prices as functions of the truncation order $N$. The right panels display the corresponding Hermite moments for multi-orders $n_1+\cdots+n_d=1,\dots,10$.
}\label{fig:fwdasian}
\end{figure}

\processdelayedfloats


\bibliographystyle{plainnat}
\bibliography{JSVM}

\end{document}